\documentclass[11pt]{amsart}
\usepackage[utf8]{inputenc}
\usepackage[nospace,noadjust]{cite}
\usepackage{relsize}
\usepackage{cancel}
\usepackage[dvipsnames]{xcolor}
\usepackage{tikz}
\usetikzlibrary{arrows.meta, positioning}
\usepackage{tikz-cd}
\usetikzlibrary{arrows.meta, calc, decorations.markings}
\usepackage{colortbl}
\usepackage{hyperref, enumitem}
\hypersetup{
	colorlinks   = true, 
	urlcolor     = blue, 
	linkcolor    = Purple, 
	citecolor   = red 
}
\usepackage[margin=2.5cm]{geometry}

\usepackage{comment}
\usepackage{cleveref,amssymb}

\usepackage{amscd}

\DeclareMathOperator{\C}{\mathcal{C}}

\newtheorem{theorem}{Theorem}[section]

\newtheorem{lemma}[theorem]{Lemma}

\newtheorem{corollary}[theorem]{Corollary}
\newtheorem{definition}[theorem]{Definition}
\newtheorem{proposition}[theorem]{Proposition}

\newtheorem{remark}[theorem]{Remark}

\usepackage{mdframed}

\usepackage{url}

\newcommand{\cL}{{\mathcal L}}

\newcommand{\F}{{\mathbb F}}

\newcommand{\Tr}{{\mathrm{Tr}}}

\newcommand{\K}{{\mathbb K}}
\newcommand{\fq}{{\mathbb F}_{q}}

\newcommand{\N}{\mathrm{N}}

\newcommand{\rk}{\mathrm{rk}}

\newcommand{\Gal}{\mathrm{Gal}}
\newcommand{\GL}{\mathrm{GL}}
\newcommand{\End}{\mathrm{End}}
\newcommand{\Aut}{\mathrm{Aut}}
\newcommand{\LL}{\mathbb{L}}

\newcommand{\drk}{\mathrm{d}_{\rk}}

\newcommand{\id}{\mathrm{id}}

\newcommand{\gcrd}{\mathrm{gcrd}}
\newcommand{\lclm}{\mathrm{lclm}}

\newcommand{\MF}{\mathcal{M}_{F}}
\newcommand{\Fix}{\mathrm{Fix}}

\newcommand{\Cen}{\mathrm{Cen}}
\newcommand{\E}{\mathbb{E}}
\newcommand{\Stab}{\mathrm{Stab}}
\newcommand{\newH}{\mathcal{S}}

\title[New Maximum Hermitian Codes and New Invariants]{New Maximum Hermitian $d$-Codes and New Invariants from Skew Polynomial Rings
}

\author[P. Santonastaso]{Paolo Santonastaso}
\address{Paolo Santonastaso, \textnormal{Department of Mathematics and Applications ``R. Caccioppoli'', University of Naples Federico II, Via Cintia, Monte Sant'Angelo, 80126 Naples, Italy }}\email{paolo.santonastaso@unina.it}

\author[Y. Zhou]{Yue Zhou}
\address{Yue Zhou, \textnormal{College of Science, National University of Defense Technology, 410073 Changsha,
		China}}
\email{yue.zhou.ovgu@gmail.com}

\date{\today}

\begin{document}

	\begin{abstract}
Maximum rank-distance codes were introduced independently by Delsarte (1978) and
Gabidulin (1985) for unrestricted matrix spaces, and their study was subsequently
extended to the classical form spaces of alternating, symmetric, and Hermitian
matrices.  For Hermitian \(d\)-codes, Schmidt (2018) established 
the upper bound and constructed optimal additive codes for all admissible parameters 
except when \(n\) and \(d\) are both even and \(4\le d\le n-2\).  We settle this final
open existence case, more generally, construct additive maximum Hermitian rank-metric
codes for every admissible length and minimum distance. Rather than working directly in the space of Hermitian matrices, we construct our objects within central simple algebras with unitary involution arising from quotients of skew polynomial rings.
We also introduce scalar-field and
Lyapunov operators in these algebras, leading to new invariants for the equivalence of additive Hermitian codes. Explicitly determining these invariants yields strong necessary conditions for Hermitian equivalence and demonstrates that our construction contains a large number of pairwise inequivalent codes.
\end{abstract}


\maketitle

\section{Introduction}

Let \(\mathcal X\) be a space of matrices over a finite field.  The rank
metric on \(\mathcal X\) is
\[
  \drk(A,B)=\rk(A-B),
\]
and a nonempty subset \(\mathcal C\subseteq\mathcal X\) is a \(d\)-code if
\(\rk(A-B)\ge d\) whenever \(A,B\in\mathcal C\) are distinct.  The basic
extremal questions are to determine the largest possible cardinality, to
construct codes attaining the relevant bound, and to classify the extremal
codes under the natural rank-metric isometries.

For unrestricted rectangular matrices, the theory began with two independent
constructions.  Delsarte used the association scheme of bilinear forms, while
Gabidulin used extension fields and linearized polynomials
\cite{delsarte1978bilinear,ga85a}.  Both constructions attain the rank-metric
Singleton bound and give the classical maximum rank-distance, or MRD, codes.
These two viewpoints---bilinear forms and linearized-polynomial evaluation---
remain fundamental in rank-metric coding theory.

A parallel extremal theory arises in matrix spaces defined by an adjoint
symmetry.  We use the term \emph{classical form spaces} for the spaces of
alternating and symmetric bilinear forms and of Hermitian sesquilinear forms;
after choosing bases, these become the spaces of alternating, symmetric, and
Hermitian matrices.  Delsarte and Goethals obtained sharp bounds and optimal
constructions in the alternating-form space
\cite{DelsarteGoethals1975}. Schmidt subsequently established corresponding results
 for symmetric forms in odd and even characteristic and for
Hermitian forms
\cite{schmidt2010symmetric,schmidt2015symmetric,schmidt2020quadratic,
Schmidt2018Hermitian}.  For some parameters of symmetric-form spaces and Hermitian-form spaces, the sharp upper bound for arbitrary codes
is not known, whereas the corresponding additive problem has a sharp bound.
Accordingly, a maximum additive \(d\)-code means an additive subgroup that
attains the sharp upper bound in its ambient form space.

Let \(\mathcal H_n(\mathbb F_{q^2}/\mathbb F_q)\) be the space of \(n\times n\)
Hermitian matrices over \(\mathbb F_{q^2}\).  Schmidt proved that every additive
Hermitian \(d\)-code \(\mathcal C\) satisfies
\begin{equation}\label{eq:intro-hermitian-bound}
  |\mathcal C|\le q^{n(n-d+1)}.
\end{equation}
He also constructed codes attaining \eqref{eq:intro-hermitian-bound} for all
admissible pairs \((n,d)\), except when \(n\) and \(d\) are both even and
\[
  4\le d\le n-2;
\]
the endpoint cases \(d=2\) and \(d=n\) were already known
\cite{Schmidt2018Hermitian}. See also \cite{Trombetti2021maximum} for another construction of codes attaining the bound in \eqref{eq:intro-hermitian-bound} when $d=2$. Thus the even--even range was the last
unresolved existence case for maximum additive \(d\)-codes among the
classical form spaces described above.

The first purpose of this paper is to close this gap.  More generally, we give
a uniform construction over every quadratic extension
\(\mathbb F_{q^s}/\mathbb F_{q^{s/2}}\), where \(s\ge2\) is even.  Set
\[
  \LL=\mathbb F_{q^n},\qquad \K=\mathbb F_q,
\]
let \(\sigma\) generate \(\operatorname{Gal}(\LL/\K)\), and consider the skew
polynomial ring
\[
  R=\LL[X;\sigma],\qquad X\alpha=\sigma(\alpha)X
  \quad(\alpha\in\LL).
\]
Choose a monic irreducible self-reciprocal polynomial
\(F(Y)\in\K[Y]\) of degree \(s\).  Since \(F(X^n)\) is central in \(R\), the
quotient
\[
  R_F:=R/RF(X^n)
\]
is a central simple algebra of degree \(n\) over its center
\(\E_F\cong\mathbb F_{q^s}\).  Self-reciprocity of \(F\) yields an involution
of the second kind
\[
  \Theta:R_F\longrightarrow R_F,
\]
whose restriction \(\varrho\) to \(\E_F\) has fixed field
\(\E_F^\varrho\cong\mathbb F_{q^{s/2}}\).  Under a suitable splitting
isomorphism \(R_F\cong M_n(\E_F)\), the fixed-point space \(\Fix(\Theta)\) is
identified with
\[
  \mathcal H_n(\E_F/\E_F^\varrho).
\]
Thus the Hermitian matrix space is realized intrinsically as the symmetric
part of a central simple algebra with unitary involution.

For \(0\le \ell\le ns/2-1\), define
\begin{equation}\label{eq:intro-code}
  \newH_\ell(F)
  :=
  \left\{
  \alpha_0+
  \sum_{i=1}^{\ell}
  \bigl(\sigma^{-i}(\alpha_i)X^{-i}+\alpha_iX^i\bigr)
  +RF(X^n)
  :
  \alpha_0,\ldots,\alpha_\ell\in\mathbb F_{q^n}
  \right\}.
\end{equation}
The paired coefficients are exactly the condition that the displayed elements
be fixed by \(\Theta\). For \(1\le d\le n\), set
\[
  \ell=\ell_d:=\frac{s(n-d+1)-2}{2}.
\]
We prove that \(\newH_{\ell_d}(F)\) has minimum rank distance \(d\) and
cardinality
\[
  q^{sn(n-d+1)/2}
  =|\mathbb F_{q^{s/2}}|^{n(n-d+1)}.
\]
It therefore attains the sharp Hermitian bound. Consequently, for any even $s$ our construction yields maximum additive Hermitian $d$-codes in
\(\mathcal H_n(\mathbb F_{q^s}/\mathbb F_{q^{s/2}})\) for every \((n,d)\), including
the previously open even--even range.

We next study equivalence.  Let \(\E/\F\) be a quadratic extension of finite
fields, let \(A^*:=(A^\varrho)^\top\), and let
\(\mathcal C,\mathcal D\subseteq\mathcal H_n(\E/\F)\) be additive codes.  The
full semilinear equivalence relation used in this paper is the following:
\(\mathcal C\) and \(\mathcal D\) are \emph{Hermitian equivalent} if
\begin{equation}\label{eq:intro-hermitian-equivalence}
  \mathcal D
  =
  \left\{
  \mu N A^\varphi N^*:
  A\in\mathcal C
  \right\}
\end{equation}
for some \(N\in\GL_n(\E)\), \(\mu\in\F^\times\), and
\(\varphi\in\Aut(\E)\), where \(A^\varphi\) is obtained by applying
\(\varphi\) entrywise.  These are precisely the semilinear rank-metric
isometries fixing the zero matrix, up to the standard translation term
\cite{wan1996geometry}.  

Our analysis of equivalence is based on two structural invariants that we introduce.  The first is the
scalar field
\[
  \mathcal F(\mathcal C)
  :=\{\lambda\in\F:\lambda\mathcal C\subseteq\mathcal C\}.
\]
If \(\Phi(A)=\mu N A^\varphi N^*\) gives an equivalence from
\(\mathcal C\) to \(\mathcal D\), then
\[
  \lambda\Phi(A)
  =\Phi\bigl(\varphi^{-1}(\lambda)A\bigr),
\]
and hence
\[
  \mathcal F(\mathcal D)=\varphi\bigl(\mathcal F(\mathcal C)\bigr).
\]
Thus the isomorphism type, and in particular the cardinality, of the scalar
field is an invariant of Hermitian equivalence.  Schmidt's trace
constructions are linear over the entire fixed field, whereas, for \(s>2\),
our codes have a proper scalar field; when \(\ell_d<ns/4\), it is exactly
\(\mathbb F_q\).  This separates the new codes from Schmidt's constructions
whenever the latter exist with the same parameters.

The second invariant is obtained from Lyapunov-type operators.  For
\(\overline b\in R_F\), put
\[
  \Delta_{\overline b}(\overline z)
  :=\overline b\,\overline z+
  \overline z\,\Theta(\overline b),
  \qquad \overline z\in\Fix(\Theta),
\]
and define
\[
  \mathcal L(\mathcal C)
  :=\{\overline b\in R_F:
  \Delta_{\overline b}(\mathcal C)\subseteq\mathcal C\}.
\]
After passage to matrices, \(\Delta_{\overline b}\) becomes the Hermitian
Lyapunov operator \(A\mapsto BA+AB^*\). The corresponding Lyapunov equations 
are classical objects in matrix analysis and control theory, and have been studied from several viewpoints; see, for instance
\cite{Lancaster1970,ByersNash1987,HewerKenney1988,LiWhite2002}.  Its
transformation under \eqref{eq:intro-hermitian-equivalence} is immediate.
Since finite-field automorphisms commute with \(\varrho\), if
\(B'=NB^\varphi N^{-1}\), then
\[
  \Delta_{B'}\bigl(\mu N A^\varphi N^*\bigr)
  =\mu N\Delta_B(A)^\varphi N^*.
\]
Consequently,
\[
  \mathcal L(\mathcal D)
  =N\mathcal L(\mathcal C)^\varphi N^{-1}.
\]
For \(n\ge3\) and \(1\le\ell<ns/4\), we determine this invariant explicitly:
\[
  \mathcal L(\newH_\ell(F))
  =\LL+\mathcal B_F,
  \qquad
  \mathcal B_F
  :=\{\lambda\in\E_F:
       \lambda+\Theta(\lambda)\in\mathbb F_q\}.
\]
The distinguished copy of \(\LL=\mathbb F_{q^n}\) can be recovered from this
additive subgroup.  Hence an equivalence between two members of the family
must normalize that field, reducing an arbitrary Hermitian isometry to a
strongly constrained comparison in a common semilinear model.

Based on this determination of the Lyapunov invariant, Section~\ref{sec:equivalence}
proves strong necessary conditions for equivalence between
\(\newH_\ell(F)\) and \(\newH_\ell(G)\).  A full Hermitian equivalence may first
apply an automorphism of \(\E=\mathbb F_{q^s}\); the conditions below are
therefore understood up to this field automorphism.  If \(\tau\) denotes the induced
automorphism of \(\K=\mathbb F_q\), write \(F^\tau\) for the polynomial
obtained by applying \(\tau\) to the coefficients of \(F\).  Assume now that $n>2$. The conclusions are as follows. 
\begin{itemize}
  \item When
  \[
    1\le\ell<\min\{n,ns/4\},
  \]
  the equivalence between
  \(\newH_\ell(F)\) and \(\newH_\ell(G)\) leads to
  \[
    F^\tau(Y)=\lambda^sG(\lambda^{-1}Y)
    \qquad\text{for some }\lambda\in\mathbb F_q^\times.
  \]

  \item When \(s>4\) and
  \[
    n<\ell<\frac{ns}{4},
  \]
  choose roots \(u_F,u_G\in\mathbb F_{q^s}\) and put
  \(w_H=u_H+u_H^{-1}\).  Then the equivalence between
  \(\newH_\ell(F)\) and \(\newH_\ell(G)\) leads to
  \[
    w_G=
    \frac{a\,\varphi(w_F)+b}{c\,\varphi(w_F)+d}
  \]
  for some \(a,b,c,d\in\mathbb F_q\) with \(ad-bc\ne0\), where
  \(\varphi\) is the field automorphism occurring in the equivalence.
\end{itemize}
Thus, up to field automorphisms, possible equivalences are confined to explicit scalar or $\mathrm{PGL}_2(\mathbb F_q)$-orbits of the associated polynomials and elements.

These restrictions give quantitative lower bounds under full Hermitian
equivalence.  Let \(N_q(s)\) be the number of monic irreducible
self-reciprocal polynomials of degree \(s\) over \(\mathbb F_q\), write
\(q=p^h\), and retain
\(\ell_d=(s(n-d+1)-2)/2\).  Define
\[
  L_2(q):=
  \begin{cases}
    q/2, & q\text{ even},\\[1ex]
    (q-1)/4, & q\equiv1\pmod4,\\[1ex]
    (q+1)/4, & q\equiv3\pmod4.
  \end{cases}
\]
The family contains at least the following numbers of pairwise Hermitian
inequivalent maximum additive codes:
\begin{itemize}
  \item
  \(
    \left\lceil L_2(q)/(2h)\right\rceil
  \)
  if \(s=2\) and \(n/2<d<n\);

  \item
  \(
    \left\lceil
      N_q(s)/(sh\,\gcd(s,q-1))
    \right\rceil
  \)
  if \(s\ge4\) is even and \(\ell_d<n\);

  \item
  \(
    \left\lceil
      N_q(s)/(sh\,q(q^2-1))
    \right\rceil
  \)
  if \(s>4\) is even and \(n<\ell_d<ns/4\).
\end{itemize}
Since \(N_q(s)=\Theta(q^{s/2})\) as $q\rightarrow \infty$, these lower bounds are respectively on the
scales \(q/h\), \(q^{s/2}/h\), and \(q^{s/2-3}/h\) for fixed admissible
\(s\).  In particular, the construction that closes the last existence gap
is far from unique.

\subsection{Organization of the paper}
Section~\ref{sec:preliminaries} reviews Hermitian rank-metric codes, skew
polynomial rings, and the arithmetic of irreducible self-reciprocal
polynomials.  Section~\ref{sec:skewpolynomialsetting} constructs the central
simple algebra \(R_F\), its involution of the second kind, and the
identification of its fixed points with a Hermitian matrix space.
Section~\ref{sec:newconstructiontotal} defines the codes \(\newH_\ell(F)\), proves
their optimality, recalls Schmidt's trace constructions, and separates the
two families by the scalar-field invariant.
Section~\ref{sec:lyapunov-invariants} introduces the Lyapunov invariant and
computes it for the new codes.  Finally, Section~\ref{sec:equivalence}
places the constructions in a common semilinear model, proves the necessary
equivalence conditions in the short and large exponent ranges, and derives
the lower bounds for full Hermitian equivalence.

\section{Preliminaries} \label{sec:preliminaries}
	
	For a field \(\K\), let \(M_n(\K)\) denote the ring of $n \times n$ matrices over $\K$.
	
	\subsection{Hermitian $d$-codes}
	Let \(\E/\F\) be a quadratic extension of finite fields, and let $\rho:\E\longrightarrow \E$
	be the nontrivial element of \(\Gal(\E/\F)\). Let \(V=V(n,\E)\) be an \(n\)-dimensional vector space over \(\E\). A map
	\[
	H:V\times V\longrightarrow \E\]
	is called a \textbf{Hermitian form} if it is \(\E\)-linear in the first variable and satisfies
	\[
	H(v,u)=\rho(H(u,v)),
	\]
	for all $u,v\in V$. Fix a basis \(\mathcal B=(\xi_{1},\dots,\xi_{n})\) of \(V\) over \(\E\). The form \(H\) is represented by the matrix
	\[
	M_{\mathcal B}(H):=\bigl(H(\xi_i,\xi_j)\bigr)_{1\le i,j\le n}\in M_n(\E).
	\]
	We have
	$
	M_{\mathcal B}(H)^{*}=M_{\mathcal B}(H),
	$
	where \(A^{*}:=(A^\rho)^{\top}\) denotes the conjugate transpose, with conjugation induced by
	\(\rho\) on \(\E\). Thus \(M_{\mathcal B}(H)\) is a \textbf{Hermitian matrix}. Conversely, every Hermitian matrix in \(M_n(\E)\) determines a unique Hermitian form on \(V\). We denote by \(\mathcal{H}(V)\) the set of all Hermitian forms on \(V\), and by \(\mathcal{H}_n(\E/\F)\) the set of all \(n\times n\) Hermitian matrices over \(\E\). Both sets carry a natural structure of \(\F\)-vector spaces. Hence, after fixing \(\mathcal B\), one obtains a natural identification between Hermitian forms on \(V\) and Hermitian matrices of order \(n\), i.e., the map
	\begin{equation}\label{eq:hermitian-identification}
		\Psi_{\mathcal B}:\mathcal H(V)\longrightarrow \mathcal{H}_n(\E/\F),\qquad
		H\longmapsto M_{\mathcal B}(H),
	\end{equation}
	is a \(\F\)-linear isomorphism. In particular,
	\[
	|\mathcal H(V)|=|\mathcal H_n(\E/\F)|=|\F|^{n^{2}},
	\]
	since a Hermitian matrix is determined by its \(n\) diagonal entries in \(\F\) and its
	\(n(n-1)/2\) entries above the diagonal in \(\E\). We define
	\[
	\rk(H):=\rk\bigl(M_{\mathcal B}(H)\bigr),
	\]
	where the matrix rank is computed over \(\E\).

	We endow the space \(\mathcal H(V)\) with the rank metric, defined by
	\[
	\drk(H_1,H_2):=\rk(H_1-H_2)
	=\rk\!\bigl(M_{\mathcal B}(H_1)-M_{\mathcal B}(H_2)\bigr),
	\]
	for all $H_1,H_2\in\mathcal H(V)$. Via \eqref{eq:hermitian-identification}, we use the same notation for the induced rank metric on \(\mathcal H_n(\E/\F)\). Thus a Hermitian rank-metric code is a nonempty subset
	\(\mathcal C\subseteq\mathcal H(V)\), or, after fixing \(\mathcal B\),
	a nonempty subset \(\mathcal C\subseteq\mathcal H_n(\E/\F)\). Its minimum distance is
	\[
	\drk(\mathcal C):=\min\{\rk(H_1-H_2): H_1,H_2\in\mathcal C,\ H_1\neq H_2\}.
	\]
	For an integer \(d\), we call \(\mathcal C\) a \textbf{Hermitian \(d\)-code} if
	\[\drk(\mathcal C)\ge d,\] equivalently if every nonzero difference of two
	distinct codewords has rank at least \(d\).
	
	An \textbf{additive Hermitian code} is a Hermitian rank-metric code that is an
	additive subgroup of \(\mathcal H_n(\E/\F)\). The corresponding extremal
	problem is to maximize \(|\mathcal C|\) among additive Hermitian \(d\)-codes in
	\(\mathcal H_n(\E/\F)\). Schmidt proved that every such code satisfies
	\begin{equation}\label{eq:hermitian-bound}
		|\mathcal C|\le |\F|^{\,n(n-d+1)},
	\end{equation}
    see \cite[Theorem 1]{Schmidt2018Hermitian}.
	Codes attaining \eqref{eq:hermitian-bound} are called \textbf{maximum
	additive Hermitian \(d\)-codes}. Schmidt also constructed such codes in
	\cite{Schmidt2018Hermitian} for every admissible pair \(1\le d\le n\), except
	when \(n\) and \(d\) are both even and \(4\le d\le n-2\). We recall these
	constructions in Section~\ref{sec:schmidt-skew}.
	
	\subsection{Central simple algebras from skew polynomial rings}
	
	We begin by recalling that a \textbf{central simple algebra} over a field \(\F\) is a finite-dimensional associative \(\F\)-algebra whose center is exactly \(\F\) and which has no nontrivial two-sided ideals. Since throughout this paper we work over finite fields, the Artin--Wedderburn theorem together with Wedderburn's little theorem implies that every central simple \(\F\)-algebra is \(\F\)-isomorphic to a full matrix algebra over \(\F\). We shall realize these matrix algebras in a way that is well suited to our constructions, namely through skew polynomial rings.

	Let \(\LL/\K\) be a finite field extension of degree \(n\). Let $\sigma$ be a generator of the Galois group $\Gal(\LL/\K)$, and consider the skew polynomial ring
	\[
	R = \LL[X; \sigma].
	\]
	The elements of $R$ are polynomials in the indeterminate $X$ with coefficients in $\LL$ written on the left of each monomial $X^i$.  
	Multiplication in $R$ is defined by the rule
	$X\alpha = \sigma(\alpha)X,$ for all $\alpha \in \LL,$
	and extended distributively. The ring $R$ is noncommutative unless $n=1$, and its center is given by $Z(R) = \K[X^n]$. There is a well-defined notion of \textbf{degree} for nonzero elements of $R$, denoted $\deg(f)$, which behaves analogously to the usual polynomial degree. Moreover, $R$ is a \emph{left Euclidean domain}.  
	That is, for any nonzero elements $f,g \in R$, there exist unique $q,r \in R$ such that
	\[
	f = qg + r,
	\qquad \deg(r) < \deg(g) \ \text{or}\ r = 0.
	\]
	Whenever $r = 0$, we say that $g$ \textbf{right-divides} $f$, and we write $g \mid_r f$; equivalently, $f$ is called a \textbf{left multiple} of $g$.  
	As a consequence, $R$ is a \emph{left principal ideal domain}. This allows one to define the \textbf{greatest common right divisor} and the \textbf{least common left multiple} of two nonzero elements $f_1,f_2 \in R$, denoted respectively by
	\(\gcrd(f_1,f_2)\) and \(\lclm(f_1,f_2)\).
	If $a = \gcrd(f_1,f_2)$ and $b = \lclm(f_1,f_2)$, then
	$
	Rf_1 + Rf_2 = Ra
	\quad\text{and}\quad
	Rf_1 \cap Rf_2 = Rb,$
	where $Rf_i$ denotes the left ideal generated by $f_i$. The reader is referred to \cite{ore1933theory,jacobson2009finite,goodearl2004introduction} for further details on the algebraic structure of skew polynomial rings. We shall obtain finite-dimensional central simple algebras by central
	reduction. Since \(Z(R)=\K[X^n]\), each irreducible polynomial
	\(F(Y)\in\K[Y]\) defines a maximal ideal of \(Z(R)\) via \(Y\mapsto X^n\).
	The corresponding quotient of \(R\) is then finite-dimensional over the
	residue field \[\frac{\K[Y]}{\K[Y]F(Y)}.\] Precisely, let $F(Y) \neq Y$ be an irreducible polynomial of $\K[Y]$ with degree $s$. Then $F(X^n)\in Z(R)=\K[X^n]$ and \(RF(X^n)\) is a two-sided ideal of \(R\). We consider the quotient ring
	\[ R_F:=\frac{R}{RF(X^n)}=
	\left\{ a_0+a_1X+\cdots +a_{ns-1}X^{ns-1} + RF(X^n) \colon a_0,\ldots,a_{ns-1} \in \LL
	\right\}. 
	\]
	
	In the following, elements of the quotient ring $R_F$ will be denoted by $\overline{a}$ and we often identify $\overline{a}$ with its canonical representative
	\[
	\overline{a} = a + RF(X^n),\]
	where $a \in R$ is the unique skew polynomial in the corresponding equivalence class having degree strictly smaller than $ns$.  
	With this identification, the degree of $\overline{a}$ is defined by
	$
	\deg(\overline{a}) := \deg(a)$. The center of $R_F$ is denoted by 
	\(\mathbb{E}_F\), and
\[
\mathbb{E}_F\cong \K[Y]/(F(Y)).
\]
Since \(F\) is irreducible of degree \(s\), the center \(\mathbb{E}_F\) is a
field with \([\mathbb{E}_F:\K]=s\).  Moreover, \(RF(X^n)\) is a maximal
two-sided ideal of \(R\), and \(R_F\) is a central simple algebra of dimension
\(n^2\) over \(\mathbb{E}_F\), hence of dimension \(n^2s\) over \(\K\); see,
for example, \cite{gomez2019computing,jacobson2009finite,sheekey2020new}.
Consequently, the Artin--Wedderburn theorem gives an \(\mathbb{E}_F\)-algebra
isomorphism
	\begin{equation} \label{eq:artin}
		\frac{R}{RF(X^n)} \cong M_n(\mathbb{E}_F).
	\end{equation}

	For an $\mathbb{E}_F$-algebra isomorphism \[\MF:R_F \stackrel{\sim}{\longrightarrow} M_n(\mathbb{E}_F),\] we can identify any element $\overline{a} \in R_F$ with its image $\MF(\overline{a})$ in $M_n(\F_{q^s})$, via the isomorphism $\MF$. Note that if \[\MF': R_F \stackrel{\sim}{\longrightarrow} M_n(\mathbb{E}_F)\] is another $\mathbb{E}_F$-algebra isomorphism then, by Skolem--Noether's Theorem, there exists \( N \in \GL_n(\mathbb{E}_F) \) such that
	$
	\MF'(\overline{a}) = N \MF(\overline{a}) N^{-1}$,
	for all \( \overline{a} \in R_F \). Therefore,
	\[\rk(\MF'(\overline{a})) = \rk(\MF(\overline{a})).\] Throughout, we often implicitly identify an element $\overline{a} \in R_F$ with its corresponding matrix in $M_n(\mathbb{E}_F)$. Accordingly, we refer to $\rk(\overline{a})$ to denote the rank of $\MF(\overline{a})$ over $\mathbb{E}_F $. Indeed, as observed above, this rank is independent of the choice of $\mathbb{E}_F$-algebra isomorphism $\MF$. The following result determines the rank of a polynomial in the matrix ring defined by $F(Y)$.

	\begin{theorem} [see \textnormal{\cite[Proposition 7]{sheekey2020new}}]\label{th:rankpolynomial_rephrased}
		For a nonzero $\overline{a}=a+RF(X^n)\in R_F$, we have
		\[
		\rk(\overline{a})
		=
		\frac{1}{s}\Big(\deg(F(X^n))-\deg\big(\gcrd(a,F(X^n))\big)\Big).
		\]
	\end{theorem}
	
	\section{Central simple algebras with unitary involution} \label{sec:skewpolynomialsetting}

	In this section, we endow the central simple algebra \(R_F\) with an involution of the second kind. This yields an intrinsic description of the Hermitian elements of \(R_F\). Through the corresponding matrix representation, we get an explicit representation of the associated Hermitian matrices.
	
	Consider a monic irreducible polynomial
	\[
	F(Y)=F_0+F_1Y+\cdots+F_{s-1}Y^{s-1}+Y^s \in \K[Y]
	\]
	of degree \(s\) with nonzero constant coefficient \(F_0\neq 0\). Since the constant term of \(F(Y)\) is nonzero, the skew polynomials \(F(X^n)\) and \(X\) are right coprime in \(R\), that is
	$
	\gcrd(F(X^n),X)=1$.
	Consequently, by the Bezout identity there exist \(u,v\in R\) such that
	\begin{equation}\label{eq:X-inverse-bezout}
		1=uX+vF(X^n).
	\end{equation}
	Recall that, in view of the algebra isomorphism in \eqref{eq:artin}, the left and right inverses of any invertible element $h+RF(X^n)$ of $R_F$ coincide, because $h+RF(X^n)$ corresponds to an invertible matrix under this isomorphism. Therefore, from \eqref{eq:X-inverse-bezout} we derive that the coset \(X+RF(X^n)\) is invertible in the ring \(R_F\), with inverse given by \(u+RF(X^n)\). For convenience, we define the notation
		\[
		X^{-1}+RF(X^n):=u+RF(X^n).
		\]

	In particular, every power \(X^i+RF(X^n)\) with \(i\ge 1\) is invertible in \(R_F\). In what follows, we denote the inverse of $X^i+RF(X^n)$ by
	\[
	X^{-i}+RF(X^n).
	\]
	
	Let \(\widehat{F}(Y)\) denote the monic reciprocal polynomial of \(F(Y)\), defined by
	\[
	\widehat{F}(Y)
	=
	F_0^{-1}Y^sF\!\left(\frac{1}{Y}\right)
	=
	F_0^{-1}\bigl(1+F_{s-1}Y+\cdots+F_1Y^{s-1}+F_0Y^s\bigr).
	\] 
	It is well known that \(F(Y)\) is irreducible over \(\F_q\) if and only if its reciprocal polynomial \(\widehat{F}(Y)\) is irreducible.
	
	\begin{theorem}  [see \textnormal{\cite{gomez2025adjoint}}] \label{th:antiwithreciprocal}
		Define the map
		\begin{equation}\label{eq:definitionanti}
			\Theta:
			\sum_{i=0}^{ns-1}a_iX^i+RF(X^n)
			\longmapsto
			\sum_{i=0}^{ns-1}\sigma^{-i}(a_i)X^{-i}+R\widehat{F}(X^n).
		\end{equation}
		Then \(\Theta\) is a ring anti-isomorphism from \(R_F\) onto \(R_{\widehat{F}}\).
		Moreover, its inverse is given by
		\begin{equation}\label{eq:inversetheta}
			\sum_{i=0}^{ns-1}a_iX^i+R\widehat{F}(X^n)
			\longmapsto
			\sum_{i=0}^{ns-1}\sigma^{-i}(a_i)X^{-i}+RF(X^n).
		\end{equation}
	\end{theorem}
	
	From now on we restrict to the case in which \(F\) is \textbf{self-reciprocal},
	meaning that it coincides with its reciprocal polynomial:
	\[
	F=\widehat F.
	\]
	Since \(F\) is irreducible, this implies that either \(\deg(F)=1\) or
	\(\deg(F)\) is even. Note that for every even integer \(s\ge 2\)
	there exist monic irreducible self-reciprocal polynomials over \(\F_q\) of
	degree \(s\). More precisely, write
	\[
	s=2^{e+1}r
	\]
	with \(e\ge0\) and \(r\) odd. Carlitz in \cite[Theorems~2--3]{Carlitz1967Reciprocal} proved that, if \(N_q(s)\) denotes the
	number of monic irreducible self-reciprocal polynomials of degree \(s\) over
	\(\F_q\), then
	\begin{equation} \label{eq:Carlitz-formula}
		N_q(s)=
		\begin{cases}
			2^{-e-1}\bigl(q^{2^e}-1\bigr),
			& \text{if \(q\) is odd and \(r=1\),}\\[1.2ex]
			2^{-e-1}\pi(r,q^{2^e}),
			& \text{otherwise,}
		\end{cases}
	\end{equation}
	where \(\pi(r,Q)\) denotes the number of monic irreducible polynomials of
	degree \(r\) over \(\F_Q\). In particular, \(N_q(s)>0\) for every even
	\(s\ge2\). Henceforth, we assume that \(s\ge2\).

	In this case, the anti-isomorphism constructed in Theorem~\ref{th:antiwithreciprocal} specializes to an anti-automorphism
	\[
	\Theta:R_F\longrightarrow R_F,
	\qquad
	\Theta\!\left(\sum_{i=0}^{ns-1}a_iX^i+RF(X^n)\right)
	=
	\sum_{i=0}^{ns-1}\sigma^{-i}(a_i)X^{-i}+RF(X^n).
	\]
	In this situation the map \(\Theta\) becomes an anti-automorphism of the ring \(R_F\) and since \(\Theta^{-1}=\Theta\), this is an involution of the ring $R_F$.
	
	Let \(\mathbb{E}_F=Z(R_F)\), and denote by
	\[
	\varrho:=\Theta|_{\mathbb{E}_F}\in \Aut_{\F_q}(\mathbb{E}_F)
	\]
	the automorphism induced by \(\Theta\) on the center.
	
	\begin{proposition}\label{prop:rho-on-center}
		The automorphism \(\varrho\) is an involution in \(\Aut(\mathbb{E}_F)\). Moreover, its fixed field is
		\[
		\mathbb{E}_F^{\varrho}
		= \F_q(X^{-n}+X^n)
		= \{\, z(X^{-n}+X^n)+RF(X^n) : z(Y)\in \K[Y] \,\},\]
		and $[\E_F:\mathbb{E}_F^{\varrho}]=2$.
	\end{proposition}
	
	\begin{proof}
		Since \(\mathbb{E}_F\) is commutative and \(\Theta^2=\id\), the map \(\varrho\) is an \(\F_q\)-automorphism of order dividing \(2\).
		By definition of \(\Theta\),
		\[
		\Theta(X^n+RF(X^n))=X^{-n}+RF(X^n).
		\]
		It follows immediately that \(\varrho\) fixes \(X^{-n}+X^n\), hence
		\[
		\F_q(X^n+X^{-n})\subseteq \mathbb{E}_F^{\varrho}.
		\]
		
		Let
		\[
		u:=X^n+RF(X^n)\in\mathbb E_F.
		\]
		The element \(u\) is not fixed by \(\varrho\).  Indeed, if
		\(u=u^{-1}\), then \(u^2=1\); since \(\mathbb E_F\) is a field, this
		would force \(u=\pm1\) (and \(u=1\) in characteristic \(2\)), contradicting
		the fact that the minimal polynomial \(F\) of \(u\) has degree \(s\ge2\).
		Thus \(\varrho\) is a nontrivial involution and
		\[
		[\mathbb E_F:\mathbb E_F^{\varrho}]=2.
		\]
		Moreover, \(u\) satisfies
		\[
		Y^2-(u+u^{-1})Y+1=0
		\]
		over \(\F_q(u+u^{-1})\).  Since \(u\notin\F_q(u+u^{-1})\), this
		extension also has degree \(2\).  Together with
		\(\F_q(u+u^{-1})\subseteq\mathbb E_F^{\varrho}\), this gives
		\[
		\mathbb E_F^{\varrho}=\F_q(u+u^{-1})
		=\F_q(X^n+X^{-n}),
		\]
		as claimed.
	\end{proof}
	
	The involution \(\Theta\) is \(\varrho\)-semilinear over the center:
	for every \(\lambda\in \mathbb{E}_F\) and every \(\overline{a}\in R_F\),
	\[
	\Theta(\lambda\overline{a})=\Theta(\overline{a})\varrho(\lambda)=\varrho(\lambda)\Theta(\overline{a}).
	\]
	Thus, once an \(\mathbb{E}_F\)-algebra isomorphism
	$
	\mathcal{M}_F:R_F\stackrel{\sim}{\longrightarrow} M_n(\mathbb{E}_F)$
	is fixed, the anti-automorphism \(\Theta\) induces an anti-automorphism of \(M_n(\mathbb{E}_F)\) which is semilinear with respect to \(\varrho\).

	\begin{theorem}\label{th:transpose-semilinear}
		There exists \(N\in \GL_n(\mathbb{E}_F)\) such that
		\begin{equation}\label{eq:semilinear-transpose}
			\mathcal{M}_F(\overline{a})^\top
			=
			N\,\mathcal{M}_F(\Theta(\overline{a}))^{\varrho}\,N^{-1},
		\end{equation}
		for all $\overline{a}\in R_F$.
	\end{theorem}
	
	\begin{proof}
		By Theorem~\ref{th:antiwithreciprocal}, specialized to the case \(F=\widehat F\), the map \(\Theta\) is an anti-automorphism of \(R_F\). Therefore
		\[
		\Phi:=\mathcal{M}_F\Theta \mathcal{M}_F^{-1}:M_n(\mathbb{E}_F)\longrightarrow M_n(\mathbb{E}_F)
		\]
		is an anti-automorphism of \(M_n(\mathbb{E}_F)\). Since \(\Theta\) restricts to \(\varrho\) on the center \(\mathbb{E}_F\), the map \(\Phi\) is \(\varrho\)-semilinear. By the Skolem--Noether theorem, there exist \(\tau\in \Aut(\mathbb{E}_F)\) and \(N\in \GL_n(\mathbb{E}_F)\) such that
		\begin{equation}\label{eq:SN-pre}
			\mathcal{M}_F(\overline{a})^\top
			=
			N\,\mathcal{M}_F(\Theta(\overline{a}))^{\tau}\,N^{-1},
		\end{equation}
		for all $\overline{a}\in R_F$.
		
		It remains to prove that \(\tau=\varrho\).  Let \(\lambda\in\mathbb{E}_F\). Since \(\mathcal{M}_F(\lambda)=\lambda I_n\) and \(\Theta(\lambda)=\varrho(\lambda)\), from \eqref{eq:SN-pre} we obtain
		\[
		\lambda I_n
		=
		N\,(\tau(\varrho(\lambda))I_n)\,N^{-1}
		=
		\tau(\varrho(\lambda))I_n.
		\]
		Thus \(\lambda=\tau(\varrho(\lambda))\) for every \(\lambda\in \mathbb{E}_F\), namely \(\tau\circ \varrho=\id_{\mathbb{E}_F}\). Since \(\varrho^2=\id\), we conclude that \(\tau=\varrho\). This proves \eqref{eq:semilinear-transpose}.
	\end{proof}
	
	The preceding theorem shows that the transpose on \(M_n(\mathbb{E}_F)\) corresponds, up to inner conjugation, to the involution \(\overline{a}\mapsto \Theta(\overline{a})\) together with the Galois action \(\varrho\) on the coefficients. In particular, the pair \[(R_F,\Theta)\] is a central simple \(\mathbb{E}_F\)-algebra endowed with a \textbf{unitary involution}, i.e.\ an involution of the second kind; see, for example, \cite[p.~20]{Knus1998book}. 
	
	We next show that, after replacing \(\mathcal{M}_F\) by a suitable conjugate representation, one may normalize the matrix \(N\) away.
	
	\begin{theorem} \label{th:normalization}
		Let \(N\in \GL_n(\mathbb{E}_F)\) be as in
		Theorem~\ref{th:transpose-semilinear}. Then there exists
		\(\lambda\in \mathbb{E}_F^\times\) satisfying
		\[
			\N_{\mathbb{E}_F/\mathbb{E}_F^{\varrho}}(\lambda)=1
		\]
		and
		\begin{equation}\label{eq:N-hermitian-scalar}
			N^\top=\lambda N^{\varrho}.
		\end{equation}
		As a consequence, there exists an \(\mathbb{E}_F\)-algebra isomorphism
		\[
		\mathcal{M}_F':R_F\stackrel{\sim}{\longrightarrow} M_n(\mathbb{E}_F)
		\]
		such that
		\begin{equation}\label{eq:normalized}
			\mathcal{M}_F'(\overline{a})^\top = \mathcal{M}_F'(\Theta(\overline{a}))^{\varrho},
		\end{equation}
		for all $\overline{a}\in R_F$.
	\end{theorem}
	
	\begin{proof}
		Set \(A=\mathcal{M}_F(\overline{a})\) and \(B=\mathcal{M}_F(\Theta(\overline{a}))\). Then \eqref{eq:semilinear-transpose} reads
		$
		A^\top=N B^{\varrho}N^{-1}$.
		Replacing \(\overline{a}\) by \(\Theta(\overline{a})\) and using \(\Theta^2=\id\), we get
		\[
		B^\top=N A^{\varrho}N^{-1}.
		\]
		Taking transpose in the latter identity and then applying \(\varrho\), we obtain
		\[
		B^{\varrho}
		=
		\bigl((N^{-1})^\top\bigr)^{\varrho}\,A^\top\,(N^\top)^{\varrho}.
		\]
		Therefore, by substituting $
		A^\top=N B^{\varrho}N^{-1}$, we get 
		\[
		B^{\varrho}
		=
		\bigl((N^{-1})^\top\bigr)^{\varrho} \, N\,B^{\varrho}\, N^{-1}\,(N^\top)^{\varrho},
		\]
		and so 
		\[
		B
		=
		(N^{-1})^\top \, N^{\varrho}\,B\, \left(N^{-1}\right)^{\varrho}\,N^\top.
		\]
		This implies that 
		\[
		\left(N^{-1}\right)^{\varrho}\,N^\top B
		=
		B\,\left(N^{-1}\right)^{\varrho}N^{\top},\]
		for all $B\in M_n(\mathbb{E}_F)$. Since the center of \(M_n(\mathbb{E}_F)\) is \(\mathbb{E}_F I_n\), there exists \(\lambda\in\mathbb{E}_F^\times\) such that \eqref{eq:N-hermitian-scalar} holds.  Transposing that identity gives
		\[
		N=\lambda(N^{\varrho})^\top,
		\]
		whereas applying \(\varrho\) gives
		\[
		(N^{\varrho})^\top=\varrho(\lambda)N.
		\]
		Consequently, \(N=\lambda\varrho(\lambda)N\), and hence
		\[
		\N_{\mathbb{E}_F/\mathbb{E}_F^{\varrho}}(\lambda)=\lambda\varrho(\lambda)=1.
		\]
		By Hilbert's Theorem~90, there exists \(\mu\in \mathbb{E}_F^\times\) such that
		\[
		\lambda=\mu^{-1}\varrho(\mu).
		\]
		Set
		$H:=\mu N$.
		Then, using \(N^\top=\lambda N^{\varrho}\), we obtain
		\[
		H^\top=\mu N^\top=\mu\lambda N^{\varrho}=\varrho(\mu)N^{\varrho}=H^{\varrho}.
		\]
		Thus \(H\) is a nondegenerate \(\varrho\)-Hermitian matrix. Over the finite field extension \(\mathbb{E}_F/\mathbb{E}_F^{\varrho}\), every nondegenerate \(\varrho\)-Hermitian form is equivalent to the standard Hermitian form, see, e.g., \cite[Theorem~7.1]{taylor1992geometry}. Therefore, there exists \(P\in \GL_n(\mathbb{E}_F)\) such that
		\begin{equation}\label{eq:hermitian-standard}
			(P^{-1})^\top\,H\,(P^{-1})^{\varrho}=I_n.
		\end{equation}
		
		Now define
		\[
		\mathcal{M}_F'(\overline{a}):=P\,\mathcal{M}_F(\overline{a})\,P^{-1},
		\qquad \overline{a}\in R_F.
		\]
		Then \(\mathcal{M}_F'\) is again an \(\mathbb{E}_F\)-algebra isomorphism. Using \eqref{eq:semilinear-transpose}, we compute
		\[
		\mathcal{M}_F'(\overline{a})^\top
		=
		(P^{-1})^\top A^\top P^\top
		=
		(P^{-1})^\top N B^{\varrho} N^{-1} P^\top.
		\]
		Since \(H=\mu N\), identity \eqref{eq:hermitian-standard} implies
		\[
		(P^{-1})^\top N B^{\varrho} N^{-1} P^\top=P^{\varrho}\,B^{\varrho}\,(P^{-1})^{\varrho}
		=
		\bigl(PBP^{-1}\bigr)^{\varrho}
		=
		\mathcal{M}_F'(\Theta(\overline{a}))^{\varrho}.
		\]
		This proves \eqref{eq:normalized}.
	\end{proof}

	Equation \eqref{eq:normalized} provides the intrinsic expression of the compatibility between transpose and the involution \(\Theta\). After a suitable normalization, the anti-automorphism \(\Theta\) becomes precisely the adjoint involution defined by the standard \(\varrho\)-Hermitian form on \(\mathbb{E}_F^n\). It follows that, whenever \(\varrho\neq \id\), the algebra with involution \((R_F,\Theta)\) is isomorphic to
	\[
	\bigl(M_n(\mathbb{E}_F),\,A\mapsto A^{*}\bigr),
	\qquad A^{*}=(A^{\varrho})^\top.
	\]

	In the rest of the paper, in view of \Cref{th:normalization}, we fix an \(\mathbb{E}_F\)-algebra isomorphism
	\[
	\mathcal{M}_F:R_F\stackrel{\sim}{\longrightarrow} M_n(\mathbb{E}_F)
	\]
	such that
	\begin{equation}\label{eq:normalized-hermitian}
		\mathcal{M}_F(\Theta(\overline{a}))=\mathcal{M}_F(\overline{a})^*
	\end{equation}
	for all \(\overline{a}\in R_F\).  The relation \eqref{eq:normalized-hermitian} immediately identifies Hermitian matrices with the fixed points of \(\Theta\).
	
	\begin{proposition}\label{prop:Hermitian-fixed-points}
		For \(\overline{a}\in R_F\), the matrix \(\mathcal{M}_F(\overline{a})\) is Hermitian if and only if
		\[
		\Theta(\overline{a})=\overline{a}.
		\]
		Equivalently,
		\[
		\mathcal{M}_F\bigl(\mathrm{Fix}(\Theta)\bigr)
		=
		\{H\in M_n(\mathbb{E}_F):H^*=H\}=\mathcal{H}_n(\E_F/\E_F^{\varrho}),
		\]
		where
		\[
		\mathrm{Fix}(\Theta):=\{\overline{a}\in R_F:\Theta(\overline{a})=\overline{a}\}.
		\]
	\end{proposition}

	We now determine \(\mathrm{Fix}(\Theta)\) explicitly inside \(R_F\). Recall that $s$ is the degree of the monic irreducible self-reciprocal polynomial $F$. Thus $s$ is even. Set
	\[
	m:=\frac{ns}{2},
	\]
	and write \(R_F\) using symmetric powers of \(X\). Since \(X+RF(X^n)\) is invertible in \(R_F\), right multiplication by \(X^{-m}\) gives an \(\LL\)-linear automorphism of \(R_F\). Therefore
	\[
	\mathcal{B}:=\{X^{-m},X^{-m+1},\dots,X^{-1},1,X,\dots,X^{m-1}\}
	\]
	defines an \(\LL\)-basis of \(R_F\). In particular, every element \(\overline{a}\in R_F\) can be written uniquely in the form
	\begin{equation}\label{eq:symmetric-expansion}
		\overline{a}=\sum_{i=-m}^{m-1}\alpha_i X^i + RF(X^n),
	\end{equation}
	with $\alpha_i\in \LL$. Recall that
	\[
	\Theta\!\left(\sum_{i=-m}^{m-1}\alpha_i X^i+RF(X^n)\right)
	=
	\sum_{i=-m}^{m-1}\sigma^{-i}(\alpha_i)X^{-i}+RF(X^n).
	\]
	Hence, after reindexing,
	\[
	\Theta(\overline{a})
	=
	\alpha_0+\sum_{i=1}^{m-1}\sigma^{-i}(\alpha_i)X^{-i}
	+\sum_{i=1}^{m}\sigma^i(\alpha_{-i})X^i+RF(X^n).
	\]
	This yields an explicit description of the fixed subspace of \(\Theta\).
	
	\begin{theorem}\label{th:explicit-hermitian-elements}
		We have
		\[
		\Fix(\Theta)
		=
		\left\{
		\sum_{i=1}^{m-1}\sigma^{-i}(\alpha_i)X^{-i}
		+\alpha_0
		+\sum_{i=1}^{m-1}\alpha_i X^i
		+RF(X^n)
		:\;
		\alpha_0,\alpha_1,\dots,\alpha_{m-1}\in\LL
		\right\}.
		\]
	\end{theorem}
	
	\begin{proof}
		Set
		\[
		D:=
		\left\{
		\sum_{i=1}^{m-1}\sigma^{-i}(\alpha_i)X^{-i}
		+\alpha_0
		+\sum_{i=1}^{m-1}\alpha_i X^i
		+RF(X^n)
		:\;
		\alpha_0,\alpha_1,\dots,\alpha_{m-1}\in\LL
		\right\}.
		\]
		We first show that \(D\subseteq \Fix(\Theta)\). Let
		$\overline{a}=a+RF(X^n)\in D,$ where
		\[
		a=
		\sum_{i=1}^{m-1}\sigma^{-i}(\alpha_i)X^{-i}
		+\alpha_0
		+\sum_{i=1}^{m-1}\alpha_i X^i
		\]
		for some \(\alpha_0,\alpha_1,\dots,\alpha_{m-1}\in\LL\). Then
		\[
		\Theta(\overline{a})
		=
		\sum_{i=1}^{m-1}\Theta(\alpha_i X^i)
		+\Theta(\alpha_0)
		+\sum_{i=1}^{m-1}\Theta\!\bigl(\sigma^{-i}(\alpha_i)X^{-i}\bigr)
		=
		\sum_{i=1}^{m-1}\sigma^{-i}(\alpha_i)X^{-i}
		+\alpha_0
		+\sum_{i=1}^{m-1}\alpha_i X^i
		=
		\overline{a},
		\]
		and therefore \(\overline{a}\in \Fix(\Theta)\). A dimension count gives the reverse inclusion. Indeed, by Proposition~\ref{prop:Hermitian-fixed-points},
		\[
		|\Fix(\Theta)|=|\mathcal{H}_n(\E_F/\E_F^{\varrho})|=|\E_F|^{\frac{n^2}{2}}.
		\]
		
		On the other hand, \(D\) has size
		\[
		|\LL|^{m}=|\K|^{\frac{n^2s}{2}}= |\E_F|^{\frac{n^2}{2}}.\]
		Since \(D\subseteq \Fix(\Theta)\) and the two sets have the same cardinality, the assertion follows.
	\end{proof}
	
	\section{New constructions} \label{sec:newconstructiontotal}
	
	In this section we provide a new family of maximum additive Hermitian
	\(d\)-codes using the central simple algebra with involution
	\((R_F,\Theta)\) studied above.  The construction below yields maximum additive Hermitian \(d\)-codes for all admissible choices
	of the parameters. We then introduce a scalar-field invariant for additive
	Hermitian codes, which enables us to distinguish the present family from the
	classical constructions of Schmidt and shows that these codes are genuinely
	new.
	
	\subsection{New constructions} \label{sec:newconstruction}
	Throughout this section, let \(\LL=\F_{q^n}\), \(\K=\F_q\), let \(s\) be a
	positive even integer, and let \(F\in\F_q[Y]\) be a monic irreducible
	self-reciprocal polynomial of degree \(s\). By \(R_F\) we denote the associated central simple algebra endowed with the involution of the second kind \(\Theta\) introduced in Section~\ref{sec:skewpolynomialsetting}. Via the $\mathbb{E}_F$-algebra isomorphism $
	\mathcal{M}_F\colon R_F\longrightarrow M_n(\mathbb{E}_F),
	$
	the \(\Theta\)-fixed subspace of \(R_F\) is identified with the space $\mathcal{H}_n(\E_F/\E_F^{\varrho})$ of $\varrho$-Hermitian matrices in \(M_n(\mathbb{E}_F)\). Our aim is to construct explicit additive families attaining the general bound for Hermitian \(d\)-codes.
	
	For \(0\le \ell \le ns/2-1\), define
	\begin{equation} \label{eq:ourcode}
		\newH_{\ell}(F):=
		\left\{
		\alpha_0+
		\sum_{i=1}^{\ell}
		\bigl(\sigma^{-i}(\alpha_i)X^{-i}+\alpha_i X^i\bigr)
		+RF(X^n)
		\;:\;
		\alpha_i\in\F_{q^n}
		\right\}
		\subseteq R_F.
	\end{equation}
		\begin{theorem}\label{th:new-construction}
		For all integers \(d\) and \(n\) with \(1\le d\le n\), the set \(\newH_{\ell}(F)\), with 
		\[
		\ell=\frac{s(n-d+1)-2}{2}
		\] 
		is a maximum additive Hermitian \(d\)-code in
		\(\mathcal{H}_n(\mathbb{E}_F/\E_F^{\varrho})\).
	\end{theorem}
	
	\begin{proof}
		By \Cref{th:explicit-hermitian-elements}, every element of \(\newH_{\ell}(F)\) is fixed by \(\Theta\), and therefore
		$
		\newH_{\ell}(F)\subseteq \Fix(\Theta)$.
		Hence, via the algebra isomorphism
		$
		\mathcal{M}_F\colon R_F\stackrel{\sim}{\longrightarrow} M_n(\mathbb{E}_F),
		$ such that \eqref{eq:normalized-hermitian} holds, 
		the set \(\newH_{\ell}(F)\) is identified with a family of Hermitian matrices, namely with a subset of \(\mathcal{H}_n(\mathbb{E}_F/\E_F^{\varrho})\). Clearly, \(\newH_{\ell}(F)\) is an \(\F_q\)-subspace of \(R_F\), with
		\[
		|\newH_{\ell}(F)|
		=
		(q^n)^{\frac{s(n-d+1)}{2}}
		=
		q^{\frac{sn(n-d+1)}{2}}.
		\]
		
		It remains to prove that \(\newH_{\ell}(F)\) is a \(d\)-code in \(\mathcal{H}_n(\mathbb{E}_F/\E_F^{\varrho})\), that is, every nonzero element has rank at least \(d\). Let
		\[
		\overline{a}=
		\alpha_0+
		\sum_{i=1}^{\frac{s(n-d+1)-2}{2}}
		\bigl(\sigma^{-i}(\alpha_i)X^{-i}+\alpha_i X^i\bigr)
		+RF(X^n)
		\in
		\newH_{\ell}(F)\setminus\{0\}.
		\]
		We claim that \(\rk(\overline{a})\ge d\). 
		First note that \(X^{\ell}+RF(X^n)\) is invertible in \(R_F\). Hence right multiplication by \(X^{\ell}+RF(X^n)\) preserves rank, and therefore
		\[
		\rk(\overline{a})=\rk(\overline{a}X^{\ell}).
		\]
		
		Moreover, no negative powers of \(X\) occur in \(\overline{a}X^{\ell}\). Thus \(\overline{a}X^{\ell}\) admits a representative \(b\in R\) such that
		\[
		\deg(b)\le \ell+\ell=2\ell=s(n-d+1)-2=s(n-d)+s-2.
		\]
		
		Now every common right divisor of \(F(X^n)\) and \(b\) is also a right divisor of \(F(X^n)\). Since \(F(Y)\) has degree \(s\), it follows from \cite[Theorem~2]{sheekey2020new} that every right divisor of \(F(X^n)\) has degree divisible by \(s\). Therefore
		\[
		\deg\bigl(\gcrd(F(X^n),b)\bigr)\le s(n-d)+s-2,
		\]
		and, because \(\deg\bigl(\gcrd(F(X^n),b)\bigr)\) is a multiple of \(s\), we in fact obtain
		\[
		\deg\bigl(\gcrd(F(X^n),b)\bigr)\le s(n-d).
		\]
		Applying \Cref{th:rankpolynomial_rephrased}, we conclude that
		\[
		\rk(\overline{a})
		=
		\rk(\overline{a}X^{\ell})
		=
		n-\frac{1}{s}\deg\bigl(\gcrd(F(X^n),b)\bigr)
		\ge
		n-\frac{1}{s}\,s(n-d)
		=
		d.
		\]
		This proves that \(\newH_{\ell}(F)\) is an additive Hermitian \(d\)-code in \(\mathcal{H}_n(\mathbb{E}_F/\E_F^{\varrho})\).
		Finally, by \eqref{eq:hermitian-bound}, every additive Hermitian \(d\)-code in \(\mathcal{H}_n(\mathbb{E}_F/\E_F^{\varrho})\) has cardinality at most
		\[
		q^{\frac{sn(n-d+1)}{2}}.
		\]
		Since \(\newH_{\ell}(F)\) attains this bound, it is maximum.
	\end{proof}

	By Proposition~\ref{prop:rho-on-center}, the extension \(\mathbb E_F/\E_F^{\varrho}\) is quadratic; equivalently, \(\E_F^{\varrho}\cong \F_{q^{s/2}}\). Thus, for every even value of \(s\) and every self-reciprocal irreducible polynomial \(F\) as above, the construction takes place over the Hermitian space associated with the quadratic extension
	\[
	\mathbb E_F/\E_F^{\varrho} \cong \mathbb F_{q^s}/\mathbb F_{q^{s/2}}.
	\]
	Consequently, Theorem~\ref{th:new-construction} gives, for every length \(n\) and every minimum distance \(d\), an explicit maximum additive Hermitian \(d\)-code of size
	\[
	|\mathbb E_F|^{\frac{n(n-d+1)}2},
	\]
	in $\mathcal{H}_n(\E_F/\E_F^{\varrho})$, attaining the Schmidt bound \eqref{eq:hermitian-bound}. In particular, no parity restriction on the pair \((n,d)\) remains. This fills the exceptional even--even range \(4\le d\le n-2\), which was left open by Schmidt's constructions. At the same time, this construction places all admissible values of \(n\) and \(d\) inside a single uniform family.
	
	In the next section we recall Schmidt's original constructions, in the cases where they exist, in a skew-polynomial form compatible with the present notation. This will allow us to compare the classical trace constructions with the new family above.

	\subsection{Schmidt's constructions in the skew-polynomial framework}
	\label{sec:schmidt-skew}
	
	We now recall Schmidt's trace constructions for maximum additive Hermitian
	rank-metric codes \cite{Schmidt2018Hermitian}, recast in a skew-polynomial
	setting. We work over the quadratic extension
	\[
	\E:=\F_{q^2},
	\qquad
	\F:=\F_q,
	\]
	and we regard
	\[
	\LL:=\F_{q^{2n}}
	\]
	as an \(n\)-dimensional vector space over \(\E\).  Let
	\[
	S:=\LL[X;\sigma],
	\qquad
	\sigma(\alpha)=\alpha^{q^2},
	\]
	and set
	\[
	S_0:=\frac{S}{S(X^n-1)}.
	\]
	Since \(\sigma^n=\id_{\LL}\), the polynomial \(X^n-1\) is central in \(S\).
	Thus, \(S_0\) is an \(\E\)-algebra isomorphic to
	\(M_n(\E)\).  Also, we can identify \(S_0\) with the algebra of \(q^2\)-linearized polynomials over
	\(\LL\) modulo \(Z^{q^{2n}}-Z\) via 
	\[
	f=\sum_{i=0}^{n-1} c_iX^i
	\mapsto 
	P_f:=\sum_{i=0}^{n-1} c_i Z^{q^{2i}}.
	\]
	To \(f\in S_0\) we associate the \(\E\)-sesquilinear form
	\begin{equation}\label{eq:schmidt-skew-form}
		H_f(u,v)
		:=
		\Tr_{\LL/\E}\bigl(P_f(u)v^q\bigr),
		\qquad u,v\in \LL.
	\end{equation}

	The Hermitian condition can now be written directly on the coefficients of
	\(f\).  Namely, if
	$
	f=\sum_{i=0}^{n-1} c_iX^i,$
	then \(H_f\) is Hermitian over \(\E/\F\) if and only if
	\begin{equation}\label{eq:schmidt-hermitian-coefficients}
		c_{n-i+1}=c_i^{q^{2n-2i+1}},
	\end{equation}
	for every \(i\in\{0,\ldots,n-1\}\), where the indices are read modulo
	\(n\); see also \cite{Longobardi2020automorphism}. Equivalently, the Hermitian
	forms are
	the fixed points of the semilinear involution
	\[
	\Theta_{\rm cyc}:S_0\longrightarrow S_0,
	\qquad
	\Theta_{\rm cyc}\!\left(\sum_{i=0}^{n-1} c_iX^i\right)
	=
	\sum_{i=0}^{n-1} c_i^{q^{2n-2i+1}}X^{n-i+1}.
	\]
	Thus
	\begin{equation}\label{eq:schmidt-fixed-hermitian}
		H_f \text{ is Hermitian over } \E/\F
		\quad\Longleftrightarrow\quad
		\Theta_{\rm cyc}(f)=f.
	\end{equation}
	
	We now recall the two constructions of maximum additive Hermitian rank-metric codes given by Schmidt \cite{Schmidt2018Hermitian}.
	
	\begin{theorem}[see \textnormal{\cite[Theorem~4]{Schmidt2018Hermitian}}]
		\label{thm:schmidt-first}
		Assume that \(1\le d\le n-1\) and that \(n\) and \(d\) have opposite parity.
		Set
		$
		t:=\frac{n-d+1}{2}$. Define
		\begin{equation}\label{eq:schmidt-first-skew}
			\mathcal K^{1}_{d}
			:=
			\left\{
			\sum_{j=1}^{t}
			\left(
			\alpha_j^qX^j+\alpha_j^{q^{2n-2j+2}}X^{n-j+1}
			\right)
			+S(X^n-1)
			:\;
			\alpha_1,\ldots,\alpha_t\in \LL
			\right\}
			\subseteq S_0.
		\end{equation}
		Then, under the identification above, \(\mathcal K^{1}_{d}\) is a maximum
		additive Hermitian \(d\)-code in \(\mathcal H_n(\E/\F)\).
	\end{theorem}

	\begin{theorem}[see \textnormal{\cite[Theorem~5]{Schmidt2018Hermitian}}]
		\label{thm:schmidt-second}
		Assume that \(n\) and \(d\) are odd and that \(1\le d\le n\).
		Set $t:=\frac{n-d}{2}$. Define
		\begin{equation}\label{eq:schmidt-second-skew}
			\mathcal K^{2}_{d}
			:=
			\left\{
			\alpha_0X^{\frac{n+1}2}
			+
			\sum_{j=1}^{t}
			\left(
			\alpha_j^qX^{\frac{n-2j+1}2}
			+
			\alpha_j^{q^{n+2j+1}}X^{\frac{n+2j+1}2}
			\right)
			+S(X^n-1)
			:
			\substack{\alpha_0\in\F_{q^n},\\
			\alpha_1,\ldots,\alpha_t\in \LL}
			\right\}
			\subseteq S_0.
		\end{equation}
		Then, under the identification above, \(\mathcal K^{2}_{d}\) is a maximum
		additive Hermitian \(d\)-code in \(\mathcal H_n(\E/\F)\).
	\end{theorem}
	
	\subsection{A scalar-field invariant}
	\label{sec:scalar-invariant-schmidt}
	
	We introduce a simple invariant which separates the codes constructed in
	\Cref{th:new-construction} from the trace constructions of Schmidt, whenever
	the latter are available over the same Hermitian space.  The invariant measures
	the largest field of scalars, inside the fixed field of the Hermitian
	involution, which preserves the code.
	
	Let \(\E/\F\) be a quadratic extension of finite fields, and let
	\(\rho\) be the nontrivial element of \(\Gal(\E/\F)\). We use the standard notion of equivalence for Hermitian rank-metric codes.  For
	\(N\in\GL_n(\E)\), \(\mu\in\F^\times\), \(\varphi\in\Aut(\E)\), and
	\(B\in\mathcal H_n(\E/\F)\), the map
	\[
	\Phi_{\mu,N,\varphi,B}:\mathcal H_n(\E/\F)\longrightarrow
	\mathcal H_n(\E/\F)
	\]
	defined by
	\[
	\Phi_{\mu,N,\varphi,B}(A)
	=
	\mu\,N A^\varphi N^*+B
	\]
	preserves the rank distance.  Here \(A^\varphi\) denotes the matrix obtained by
	applying \(\varphi\) entrywise.  
	
	Conversely, every semilinear isometry of the Hermitian space preserving the rank-metric is of this form; see
	\cite{wan1996geometry}. Accordingly, two codes
	\(\mathcal C,\mathcal D\subseteq\mathcal H_n(\E/\F)\) are said to be
	\textbf{Hermitian equivalent} if
	\[
	\mathcal D=\Phi_{\mu,N,\varphi,B}(\mathcal C)
	\]
	for some \(N\in\GL_n(\E)\), \(\mu\in\F^\times\),
	\(\varphi\in\Aut(\E)\), and \(B\in\mathcal H_n(\E/\F)\).  For additive codes
	the translation term may be omitted; equivalently,
	\begin{equation}\label{eq:hermitian-equivalence}
		\mathcal D
		=\Phi_{\mu,N,\varphi}(\mathcal C)=
		\left\{
		\mu\,N A^\varphi N^*
		:
		A\in\mathcal C
		\right\}.
	\end{equation}
	When \(\varphi=\id_{\mathbb E}\), we say that \(\mathcal C\) and
	\(\mathcal D\) are \textbf{linearly Hermitian equivalent}.
	
	The following elementary invariant records the largest subfield of the fixed
	field over which an additive Hermitian code remains linear.
	\begin{definition}
		Let
		$
		\mathcal C\subseteq \mathcal H_n(\E/\F)
		$
		be a nonzero additive code.  We define
		\[
		\mathcal{F}(\mathcal C)
		:=
		\{\lambda\in \F:\lambda\mathcal C\subseteq \mathcal C\}.
		\]
	\end{definition}
	
	\begin{lemma}\label{lem:scalar-field}
		For every nonzero additive code
		$
		\mathcal C\subseteq \mathcal H_n(\E/\F),
		$
		the set \(\mathcal{F}(\mathcal C)\) is a subfield of \(\F\).
	\end{lemma}
	
	\begin{proof}
		Clearly \(0,1\in \mathcal{F}(\mathcal C)\).  If
		\(\lambda,\eta\in \mathcal{F}(\mathcal C)\), then, since \(\mathcal C\) is
		additive,
		\[
		(\lambda+\eta)\mathcal C\subseteq \mathcal C,
		\qquad
		\lambda\eta\,\mathcal C\subseteq \mathcal C.
		\]
		Moreover, \(-\lambda\in\mathcal F(\mathcal C)\) whenever \(\lambda\in\mathcal F(\mathcal C)\).  If \(\lambda\neq0\), then the inclusion
		\(\lambda\mathcal C\subseteq \mathcal C\) is an equality because \(\mathcal C\) is finite, and hence
		\(\lambda^{-1}\mathcal C=\mathcal C\).  Thus
		\(\mathcal{F}(\mathcal C)\) is a subfield of \(\F\).
	\end{proof}
	
	\begin{proposition}\label{prop:scalar-field-invariant}
		Let
		$
		\mathcal C,\mathcal D\subseteq \mathcal H_n(\E/\F)
		$
		be nonzero additive codes.  Suppose that there exist
		$
		N\in\GL_n(\E),
		\mu\in \F^\times,$ and $
		\varphi\in\Aut(\E),
		$
		such that
		\[
		\mathcal D
		=
		\Phi_{\mu,N,\varphi}(\mathcal C).
		\]
		Then
		\[
		\mathcal F(\mathcal D)
		=
		\varphi\bigl(\mathcal F(\mathcal C)\bigr).
		\]
		In particular, the cardinality of \(\mathcal F(\mathcal C)\) is invariant
		under Hermitian equivalence.
	\end{proposition}
	
	\begin{proof}
		Let \(\lambda\in \F\).  For every \(A\in\mathcal C\), we
		have
		\[
		\lambda \Phi_{\mu,N,\varphi}(A)
		=
		\mu\,N\bigl((\varphi^{-1}(\lambda)A)^\varphi\bigr)N^*
		=
		\Phi_{\mu,N,\varphi}\bigl(\varphi^{-1}(\lambda)A\bigr).
		\]
		Thus
		\[
		\lambda\mathcal D\subseteq\mathcal D
		\quad\Longleftrightarrow\quad
		\varphi^{-1}(\lambda)\mathcal C\subseteq\mathcal C.
		\]
		Equivalently,
		\[
		\lambda\in\mathcal F(\mathcal D)
		\quad\Longleftrightarrow\quad
		\varphi^{-1}(\lambda)\in\mathcal F(\mathcal C).
		\]
		Therefore
		$
		\mathcal F(\mathcal D)
		=
		\varphi\bigl(\mathcal F(\mathcal C)\bigr),
		$ proving the assertion.
	\end{proof}
	
	We first compute the invariant for Schmidt's constructions. Here, \(\E=\F_{q^2}\), \(\F=\F_q\),
	\(\LL=\F_{q^{2n}}\), \(S=\LL[X;\sigma]\), with
	\(\sigma(\alpha)=\alpha^{q^2}\), and \(S_0=S/S(X^n-1)\).  Via the
	identification described above, the codes
	\[
	\mathcal K^1_d\subseteq S_0 \cong M_n(\F_{q^2})
	\qquad\text{and}\qquad
	\mathcal K^2_d\subseteq S_0 \cong M_n(\F_{q^2})
	\]
	defined in \eqref{eq:schmidt-first-skew} and
	\eqref{eq:schmidt-second-skew} are regarded as additive codes in
	\(\mathcal H_n(\F_{q^2}/\fq)\).
	
	\begin{proposition}\label{prop:schmidt-scalar-field}
		Let \(\mathcal C\subseteq \mathcal H_n(\F_{q^2}/\fq)\) be either the code
		corresponding to \(\mathcal K^1_d\) in \eqref{eq:schmidt-first-skew}, when
		\(1\le d\le n-1\) and \(n,d\) have opposite parity, or the code corresponding
		to \(\mathcal K^2_d\) in \eqref{eq:schmidt-second-skew}, when \(n,d\) are odd.
		Then
		\[
		\mathcal F(\mathcal C)=\fq.
		\]
	\end{proposition}
	
	\begin{proof}
		By definition, \(\mathcal F(\mathcal C)\subseteq \fq\).  Thus, it is enough to prove that \(\fq\) acts on the corresponding code in \(S_0\). Assume first that \(\mathcal C\) corresponds to
		\(\mathcal K^1_d\).  For \(\lambda\in \F_q\), we have
		\[
		\lambda
		\left(
		a_j^qX^j+a_j^{q^{2n-2j+2}}X^{n-j+1}
		\right)
		=
		(\lambda a_j)^qX^j
		+
		(\lambda a_j)^{q^{2n-2j+2}}X^{n-j+1},
		\]
		since \(\lambda^q=\lambda\). This shows that
		$
		\lambda\mathcal K^1_d\subseteq \mathcal K^1_d$. The same calculation gives
		$
		\lambda\mathcal K^2_d\subseteq \mathcal K^2_d.
		$ In either case \(\fq\subseteq\mathcal F(\mathcal C)\), and the reverse
		inclusion was already observed. The assertion follows.
	\end{proof}

	We now consider the family constructed in \Cref{th:new-construction}.  We use
	the notation of \Cref{sec:newconstruction}.  Thus
	$
	\LL=\F_{q^n},$ $\K=\F_q,$ and \(F\in\F_q[Y]\) is a monic, irreducible, self-reciprocal polynomial of even
	degree \(s\). We set
	\(R_F:=R/RF(X^n)\). The center of \(R_F\) is
	\[
	\E=\E_F:=Z(R_F)\simeq \F_{q^s}.
	\]
	The involution \(\Theta\) induces on \(\E_F\) an involution \(\varrho\), whose
	fixed field is
	\[
	\F:=\E_F^\varrho\simeq \F_{q^{s/2}}.
	\]
	Thus the ambient Hermitian space associated with \(F\) is
	$
	\mathcal H_n(\E_F/\E_F^\varrho)$.
	
	\begin{proposition}\label{prop:our-scalar-field-proper}
		Assume that \(2\le d\le n\), and put
		$
		\ell:=\frac{s(n-d+1)-2}{2}$.
		Let
		\[
		\C:=\newH_{\ell}(F)
		=
		\left\{
		\alpha_0+
		\sum_{i=1}^{\ell}
		\left(
		\sigma^{-i}(\alpha_i)X^{-i}+\alpha_i X^i
		\right)
		+RF(X^n)
		:
		\alpha_0,\ldots,\alpha_\ell\in\LL
		\right\}\subseteq \mathcal{H}_n(\E/\F).
		\]
		Then
		\[
		\mathcal F(\C)\subsetneq \F
		\]
		whenever \(s>2\).  Moreover, if
		$
		\ell<\frac{ns}{4},$
		then
		\[
		\mathcal F(\C)=\F_q.
		\]
	\end{proposition}
	
	\begin{proof}
		We have $\mathcal{F}(\C)=\{\lambda \in \F \colon \lambda \C \subseteq \C\}$. Every element \(\lambda\in\F\) has a unique expression
		\begin{equation}\label{eq:fixed-field-eta-expansion}
			\lambda
			=
			\sum_{h=0}^{s/2-1}\lambda_h(X^{-n}+X^n)^h+RF(X^n),
		\end{equation}
		with $\lambda_h\in\F_q$.
		Let $m:=\frac{ns}{2}$. By \Cref{th:explicit-hermitian-elements}, every element of \(\Fix(\Theta)\)
		has a unique symmetric expansion with exponents in the range $[-m+1,m-1]$. In this expansion, the elements of \(\C=\newH_\ell(F)\) are precisely
		those fixed elements whose symmetric support is contained in the interval $[-\ell,\ell]$. We first prove that \(\mathcal F(\C)\subsetneq \F\) when \(s>2\).  Set
		\[
		\eta:=X^{-n}+X^n+RF(X^n)\in\F.
		\]
		It is enough to show that \(\eta\mathcal C\nsubseteq\mathcal C\).
		
		If \(\ell<n\), then \(1+RF(X^n)\in\C\), whereas
		\[
		\eta(1+RF(X^n))=X^{-n}+X^n+RF(X^n)
		\]
		has symmetric support \(\{-n,n\}\), which is not contained in
		\([-\ell,\ell]\).  Hence \(\eta\notin\mathcal F(\C)\).  Assume now that \(\ell\ge n\).  We claim that there exists an integer \(i\)
		such that
		\[
		0<i\le \ell,
		\qquad
		\ell<i+n<m.
		\]
		Indeed, if \(\ell<m-n\), we may take \(i=\ell\).  If \(\ell\ge m-n\), we take
		\[
		i:=m-n-1.
		\]
		This integer is positive since \(s>2\) and \(n\ge2\). Moreover,
		\(i\le\ell\), and
		\[
		i+n=m-1.
		\]
		On the other hand, since \(2\le d\le n\), we have
		\[
		\ell=\frac{s(n-d+1)-2}{2}
		\le
		\frac{s(n-1)-2}{2}
		=
		m-\frac{s}{2}-1
		<
		m-1.
		\]
		Thus \(\ell<i+n<m\), as required. Now
		\[
		\overline{u}_i:=X^i+X^{-i}+RF(X^n)
		\]
		belongs to \(\C\).  Multiplying by \(X^n+X^{-n}+RF(X^n)\), we obtain
		\[
		X^{i+n}+X^{i-n}+X^{n-i}+X^{-i-n}+RF(X^n).
		\]
		The exponent \(i+n\) satisfies
		$
		\ell<i+n<m$. Therefore, $(X^{-n}+X^n)\overline{u}_i \notin \C$ and consequently
		\[
		\mathcal F(\C)\subsetneq \F
		\]
		for \(s>2\).
		
		We now assume that
		$
		\ell<\frac{ns}{4}=\frac m2$
		and prove that
		$
		\mathcal F(\C)=\F_q.
		$
		The inclusion \(\F_q\subseteq\mathcal F(\C)\) is immediate, since \(\C\) is an
		\(\F_q\)-subspace. Conversely, let
		\[
		\lambda\in\mathcal F(\C)\subseteq\F,
		\]
		and write \(\lambda\) as in \eqref{eq:fixed-field-eta-expansion}.  Let \(t\)
		be the largest index such that \(\lambda_t\neq0\).  We prove that \(t=0\).
		
		Suppose, to the contrary, that \(t>0\).  Since \(1+RF(X^n)\in\C\) and
		\(\lambda\C\subseteq\C\), we have
		\(\lambda\in\C\).
		In the expansion
		\[
		\lambda=\sum_{h=0}^{s/2-1}\lambda_h(X^n+X^{-n})^h+RF(X^n),
		\]
		the highest positive exponent is \(tn\), and its coefficient is
		\(\lambda_t\neq0\).  Thus, since \(\lambda\in\C\), we get
		$
		tn\le \ell$. Now consider
		\[
		\overline{u}_\ell:=X^\ell+X^{-\ell}+RF(X^n)\in\C.
		\]
		Since \(\lambda\in\mathcal F(\C)\), we must have
		$
		\lambda \overline{u}_\ell\in\C$.
		However, the product \(\lambda \overline{u}_\ell\) has a nonzero component at the exponent $
		\ell+tn$. Since \(t>0\),
		we have
		$
		\ell+tn>\ell$.
		Moreover, using \(tn\le\ell\) and \(\ell<m/2\), we get
		\[
		\ell+tn\le 2\ell<m.
		\]
		Thus \(\ell+tn\) lies in the range $[-m+1,m-1]$ but outside the interval \([-\ell,\ell]\).  This contradicts \(\lambda \overline{u}_\ell\in\C\).  Hence
		\(t=0\), and therefore \(\lambda\in\F_q\). This proves the claim.
	\end{proof}

	We can now separate the new family from Schmidt's trace constructions.
	
	\begin{theorem}\label{thm:not-equivalent-to-schmidt}
		Assume that \(s>2\) and \(2\le d\le n\). Let
		\[
		\C=\newH_{\ell}(F)
		\subseteq \mathcal H_n(\E/\F),
		\qquad
		\ell=\frac{s(n-d+1)-2}{2},
		\]
		be the maximum additive Hermitian \(d\)-code constructed in
		\Cref{th:new-construction}, where
		$\E=\E_F \cong \F_{q^s},$
		$\F=\E_F^\varrho=\F_{q^{s/2}}$. 
		Let
		\[
		\mathcal C'\subseteq \mathcal H_n(\E/\F)
		\]
		be a Schmidt trace construction from Theorem~\ref{thm:schmidt-first} or
		Theorem~\ref{thm:schmidt-second}, whichever applies to the parameters \(n\)
		and \(d\), realized in the Hermitian space \(\mathcal H_n(\E/\F)\). Then
		\(\C\) and \(\C'\) are not Hermitian equivalent.
	\end{theorem}
	
	\begin{proof}
		Apply Proposition~\ref{prop:schmidt-scalar-field} with the base field there replaced by \(\F\) (equivalently, with \(q\) there replaced by \(q^{s/2}\)).  The code \(\mathcal C'\) then satisfies
		\[
		\mathcal F(\mathcal C')=\F.\]
		On the other hand, by Proposition~\ref{prop:our-scalar-field-proper}, the code
		\(\C\) satisfies
		\[
		\mathcal F(\C)\subsetneq \F.
		\]
		Since the cardinality of \(\mathcal F(\,\cdot\,)\) is invariant under
		Hermitian equivalence by \Cref{prop:scalar-field-invariant}, the two codes
		cannot be Hermitian equivalent.
	\end{proof}
	
	\begin{remark}
		The restriction \(s>2\) is precisely what makes the scalar-field invariant
		effective for the separation argument above.  Indeed, when \(s=2\), one has
		\[
		\E_F\simeq \F_{q^2}
		\qquad\text{and}\qquad
		\F\simeq \F_q,
		\]
		independently of the choice of \(F\).  Hence this invariant coincides with the
		corresponding invariant of the classical trace construction and cannot, by
		itself, distinguish the quadratic members of our family from the previously
		known construction.
		
		This limitation is only for the invariant, not for the construction.
		The case \(s=2\) is treated by a different argument in
		Subsection~\ref{subsec:enumerative-consequences}, where we use the equivalence
		criterion to prove that, in the range \(n/2<d<n\), the codes
		\[
		\newH_{n-d}(F),
		\qquad
		F\in\mathcal I^{\mathrm{sr}}_{q,2},
		\]
		give rise, under full Hermitian equivalence, to at least
		\[
		\left\lceil\frac{L_2(q)}{2h}\right\rceil
		\]
		pairwise inequivalent codes, where \(q=p^h\); see
		Corollary~\ref{cor:enumerative-Hermitian-combined}.  Thus, already in the
		quadratic case, the construction produces many new equivalence classes,
		although the scalar-field invariant alone does not separate them from the
		classical construction.
		
		Moreover, the case \(d=1\) is excluded from the separation statement because
		then
		\[
		\ell=\frac{sn-2}{2},
		\]
		and the construction gives the whole Hermitian space.
	\end{remark}

	\section{Lyapunov maps and structural invariants}
	\label{sec:lyapunov-invariants}
	
	In this section, we introduce Lyapunov maps as structural tools for the
	central simple algebra \(R_F\) endowed with the involution of the second kind
	\(\Theta\).  For any additive Hermitian code in this setting, these maps give
	rise to an
	additive subgroup consisting of all matrices whose associated Lyapunov
	operators stabilize the code. We prove that this subgroup is invariant under
	Hermitian semilinear equivalence. We then compute it explicitly for the new
	family of codes constructed in the previous section. The resulting formula
	provides a key structural invariant for the equivalence analysis carried out
	in the next section.
	
	\subsection{The Lyapunov invariant}
	
	Let \(\E/\F\) be a quadratic extension of finite fields, and let
	\(\varrho\) be the nontrivial element of \(\Gal(\E/\F)\). Consider the space of Hermitian matrices 
	$
	\mathcal H_n(\E/\F)$.

	For \(B\in M_n(\E)\), define the \(\F\)-linear map
	\[
	\Delta_B:M_n(\E)\longrightarrow M_n(\E),
	\qquad
	\Delta_B(A):=BA+AB^*,
	\]
	where \(B^*:=(B^\varrho)^\top\), consistently with the notation of Section~\ref{sec:preliminaries}.
	
	If \(A\in\mathcal H_n(\E/\F)\), then
	\[
	\Delta_B(A)^*
	=
	(BA+AB^*)^*
	=
	A^*B^*+BA^*
	=
	AB^*+BA
	=
	\Delta_B(A).
	\]
	Thus \(\Delta_B\) maps \(\mathcal H_n(\E/\F)\) into itself.  This is the
	Hermitian analogue of the classical Lyapunov map
	$
	A\longmapsto BA+AB^*,$
	in the context of algebras with involution of the second kind.
	
	\begin{definition}\label{def:lyapunov-invariant}
		Let
		$
		\mathcal C\subseteq\mathcal H_n(\E/\F)
		$
		be an additive code.  We define the \textbf{Lyapunov invariant} of $\C$ by
		\[
		\mathcal L(\mathcal C)
		:=
		\{B\in M_n(\E):\Delta_B(\mathcal C)\subseteq\mathcal C\}.
		\]
	\end{definition}
	
	It is straightforward to show that the set \(\mathcal L(\mathcal C)\) is an additive subgroup of \(M_n(\E)\).

	\begin{proposition}\label{prop:lyapunov-invariant-matrix}
		Let
		$
		\mathcal C,\mathcal D\subseteq\mathcal H_n(\E/\F)
		$
		be additive Hermitian codes.  Suppose that
		\[
		\mathcal D
		=
		\left\{
		\mu\,N A^\varphi N^*
		:
		A\in\mathcal C
		\right\}
		\]
		for some
		$
		N\in\GL_n(\E),
		\mu\in\F^\times,
		\varphi\in\Aut(\E).
		$
		Then
		\[
		\mathcal L(\mathcal D)
		=
		\left\{
		N B^\varphi N^{-1}
		:
		B\in\mathcal L(\mathcal C)
		\right\}=N \cL(\C)^{\varphi} N^{-1}.
		\]
		In particular, \(\lvert\mathcal L(\mathcal C)\rvert=\lvert\mathcal L(\mathcal D)\rvert\).
	\end{proposition}
	
	\begin{proof}
		Let \(B\in M_n(\E)\), and set
		$
		B':=N B^\varphi N^{-1}$.
		For \(A\in\mathcal C\), we compute
		\[
		\begin{aligned}
			\Delta_{B'}(\mu N A^\varphi N^*)
			&=
			B'(\mu N A^\varphi N^*)
			+
			(\mu N A^\varphi N^*)B'^*  \\
			&=
			\mu N B^\varphi A^\varphi N^*
			+
			\mu N A^\varphi (B^*)^\varphi N^* \\
			&=
			\mu N (BA+AB^*)^\varphi N^* \\
			&=
			\mu N\Delta_B(A)^\varphi N^*.
		\end{aligned}
		\]
		Hence
		\[
		\Delta_{B'}(\mathcal D)\subseteq\mathcal D
		\quad\Longleftrightarrow\quad
		\Delta_B(\mathcal C)\subseteq\mathcal C.
		\]
		Therefore
		\[
		B\in\mathcal L(\mathcal C)
		\quad\Longleftrightarrow\quad
		N B^\varphi N^{-1}\in\mathcal L(\mathcal D),
		\]
		which proves the asserted equality.  
	\end{proof}

	\subsection{Computation of the Lyapunov invariant for \(\newH_{\ell}(F)\)}
	\label{subsec:lyapunov-computation}
		
	We keep the notation of \Cref{sec:newconstruction}. Thus
	\[
	\LL=\F_{q^n},
	\qquad
	\K=\F_q,
	\qquad
	R=\LL[X;\sigma],
	\qquad
	\langle\sigma\rangle=\Gal(\LL/\K).
	\]
	Let \(F\in\F_q[Y]\) be monic, irreducible, and self-reciprocal of even degree
	\(s\), and set
	\[
	R_F:=R/RF(X^n).
	\]
	
	The center of \(R_F\) is
	\[
	\E=\E_F:=Z(R_F)\simeq \F_{q^s}.
	\]
	The involution \(\Theta\) induces an involution \(\varrho\) on \(\E\), and
	we write
	\[
	\F:=\E^\varrho\simeq \F_{q^{s/2}}.
	\]
	Thus \(\Fix(\Theta)\subseteq R_F\) is identified with the Hermitian space
	\[
	\mathcal H_n(\E/\F).
	\]
	
	Throughout this subsection, we compute the preimage in \(R_F\) of the
	Lyapunov invariant under the normalized representation \(\mathcal M_F\). Equivalently, for a code \(\mathcal C\subseteq\Fix(\Theta)\), we write
	\[
	\mathcal L(\mathcal C)
	:=
	\{\overline b\in R_F:
	\overline b\,\overline{a}+\overline{a}\Theta(\overline b)\in\mathcal C
	\text{ for all }\overline{a}\in\mathcal C\}.
	\]
	Under \(\mathcal M_F\), this agrees with the matrix invariant introduced
	above.

	We now compute the Lyapunov invariant of the Hermitian code \(\newH_{\ell}(F)\).  The result is stated in Theorem~\ref{thm:lyapunov-invariant-computation}.  As in Section~\ref{sec:newconstruction}, set
	\[
	m:=\frac{ns}{2}.
	\]
	Since \(F\) is a monic irreducible self-reciprocal polynomial of degree \(s\), the integer \(s\) is even, and hence \(m\) is an integer.  Every element \(\overline{b}\in R_F\) admits a
	unique expansion
	\[
	\overline{b}=\sum_{i=-m}^{m-1} \beta_iX^i+RF(X^n),
	\qquad \beta_i\in\LL,
	\]
	with respect to the interval \([-m,m-1]\). 
	
	We shall repeatedly use
	the following elementary fact without further mention: if \(I\subseteq\mathbb Z\)
	is an interval of length \(<ns\), then the monomials \(X^i+RF(X^n)\), \(i\in I\), are
	\(\LL\)-linearly independent in \(R_F\).
	
	Define
	\[
	\mathcal B_F
	:=
	\{\lambda\in\E:\lambda+\Theta(\lambda)\in\K\}.
	\]
	
	\begin{proposition}\label{prop:easy-inclusion-lyapunov}
		For every \(0\le \ell \le ns/2-1\), one has
		\[
		\LL+\mathcal B_F\subseteq\mathcal L(\newH_{\ell}(F)).
		\]
	\end{proposition}
	
	\begin{proof}
		Let \(\alpha \in\LL\).  Since \(\Theta(\alpha +RF(X^n))=\alpha +RF(X^n)\), for every
		\[
		\overline{a}=\alpha_0+
		\sum_{i=1}^{\ell}
		\left(
		\sigma^{-i}(\alpha_i)X^{-i}+\alpha_iX^i
		\right)+RF(X^n)
		\in\newH_{\ell}(F)
		\]
		we have
		\[
		\Delta_{\alpha}(\overline{a})=\alpha \overline{a}+\overline{a} \alpha.
		\]
		Hence, the coefficient at \(X^i\) is
		\[
		\alpha\alpha_i+\alpha_i\sigma^i(\alpha),
		\]
		while the coefficient at \(X^{-i}\) is
		\[
		\alpha\sigma^{-i}(\alpha_i)+\sigma^{-i}(\alpha_i)\sigma^{-i}(\alpha)
		=
		\sigma^{-i}\!\left(\alpha\alpha_i+\alpha_i\sigma^i(\alpha)\right).
		\]
		Thus \(\Delta_{\alpha}(\overline{a})\in\newH_{\ell}(F)\), and hence
		\(\LL\subseteq\mathcal L(\newH_{\ell}(F))\).
		
		Now let \(\lambda\in\mathcal B_F\).  Since \(\lambda\) is central, for every
		\(\overline{a}\in\newH_{\ell}(F)\) we have
		\[
		\Delta_{\lambda}(\overline{a})
		=
		\lambda\overline{a}+\overline{a}\Theta(\lambda)
		=
		(\lambda+\Theta(\lambda))\overline{a}.
		\]
		By definition, \(\lambda+\Theta(\lambda)\in\K=\F_q\).  Since \(\newH_{\ell}(F)\) is an
		\(\F_q\)-space, it follows that
		$
		\Delta_{\lambda}(\overline{a})\in\newH_{\ell}(F)$.
		Thus
		\[
		\mathcal B_F\subseteq\mathcal L(\newH_{\ell}(F)).
		\]
		Since \(\mathcal L(\newH_{\ell}(F))\) is additive, the assertion follows.
	\end{proof}
	
	We now prove the reverse inclusion in the regime $\ell < m/2$. 
	
	\begin{theorem}\label{thm:lyapunov-invariant-computation}
		Assume that $1\leq  \ell < m/2$ and $n\geq 3$. Then
		\[
		\mathcal L(\newH_{\ell}(F))=\LL+\mathcal B_F.
		\]
	\end{theorem}

	\begin{proof}
		By Proposition \ref{prop:easy-inclusion-lyapunov}, it remains to prove the reverse
		inclusion. 
				
		Let
		\[
		\overline b
		=
		\sum_{i=-m}^{m-1}\beta_iX^i+RF(X^n)
		\in
		\mathcal L(\newH_\ell(F)),
		\]
		with $\beta_i\in\LL$, be written with respect to the basis
		\[
		X^{-m},X^{-m+1},\ldots,X^{-1},1,X,\ldots,X^{m-1}
		\]
		of $R_F$. Write also
		\[
		\Theta(\overline b)
		=
		\sum_{i=-m}^{m-1}\gamma_iX^i+RF(X^n),
		\]
		with $\gamma_i\in\LL$, with respect to the same basis.
		
		We derive the required restrictions on the coefficients \(\beta_i\) in four steps.
		\begin{itemize}
			\item[Step~1.] Derive coefficient relation imposed by the condition
			\(1+RF(X^n)\in\newH_\ell(F)\), and use it to prove that all noncentral
			coefficients with exponents outside the interval \([-\ell,\ell]\) vanish.
			\item[Step~2.] Consider the remaining central coefficients outside
			\([-\ell,\ell]\) and show that they belong to \(\K\), thereby identifying the corresponding contribution as central.
			\item[Step~3.] Treat the coefficients with exponents in \([-\ell,\ell]\) by a
			descending induction to show that, modulo central terms, all nonconstant
			coefficients vanish.
			\item[Step 4.] Determine the possible central summand and prove that it belongs
			to \(\mathcal B_F\), which yields the desired inclusion.
		\end{itemize}

		\noindent\textbf{Step 1: On the noncentral coefficients outside the interval \([-\ell,\ell]\).}
		
		Since \(1+RF(X^n)\in\newH_\ell(F)\), we have
		\[
		\overline b+\Theta(\overline b)
		=
		\Delta_{\overline b}(1+RF(X^n))
		\in\newH_\ell(F).
		\]
		Therefore, for every \(i\in[-m,m-1]\) with \(|i|>\ell\), the coefficient of
		\(X^i\) in \(\overline b+\Theta(\overline b)\) is zero, that is,
		\begin{equation}\label{eq:outside-constant-relation}
			\beta_i+\gamma_i=0.
		\end{equation}
		
		In this step, we eliminate all noncentral coefficients outside the interval $[-\ell,\ell]$. Let
		\[
		|i|>\ell,
		\qquad
		i\not\equiv0\pmod n.
		\]
		For every \(\alpha\in\LL\), the constant element \(\alpha+RF(X^n)\) belongs to
		\(\newH_\ell(F)\). Hence
		\[
		\Delta_{\overline b}(\alpha+RF(X^n))=\overline{b}\alpha+\alpha \Theta(\overline{b})\in\newH_\ell(F).
		\]
		The coefficient of \(X^i+RF(X^n)\) in \(\Delta_{\overline b}(\alpha)\) is
		\[
		\beta_i\sigma^i(\alpha)+\alpha\gamma_i.
		\]
		Since \(X^i\) lies outside the interval $[-\ell,\ell]$, this coefficient is zero for every
		\(\alpha\in\LL\). Using \eqref{eq:outside-constant-relation}, we obtain
		\[
		\beta_i\bigl(\sigma^i(\alpha)-\alpha\bigr)=0,
		\]
		for every \(\alpha\in\LL\). Since \(i\not\equiv0\pmod n\), one has
		\(\sigma^i\ne\id_{\LL}\). Thus \(\beta_i=0\). Then
		\eqref{eq:outside-constant-relation} also gives \(\gamma_i=0\). Consequently,
		all noncentral coefficients of \(\overline b\) outside \([-\ell,\ell]\)
		vanish.
		
		\noindent\textbf{Step 2: On the central coefficients outside the interval \([-\ell,\ell]\).}
		
		Write \(s=2r\), so that \(m=nr\), and put \(Y=X^n\). Since \(F\) is monic, self-reciprocal and has
		constant term \(1\), the relation \(F(Y)=0\), after multiplication by
		\(Y^{-r}\), gives
		\begin{equation}\label{eq:boundary-relation-lyap-final}
			X^m
			=
			-X^{-m}+w
			\qquad\text{in }R_F,
		\end{equation}
		where
		\[
		w\in
		\sum_{\substack{-m<\nu n<m}}\K X^{\nu n}.
		\]
		Thus rewriting the boundary term \(X^m\) only produces central powers with
		coefficients in \(\K\). 
		
		Let \(t\) be a positive multiple of \(n\) such that
		\[
		\ell<t<m.
		\]
		For \(\alpha\in\LL\), set
		\[
		Z_1(\alpha):=\sigma^{-1}(\alpha)X^{-1}+\alpha X+RF(X^n)
		\in\newH_\ell(F).
		\]
		Since
		\[
		\Delta_{\overline b}(Z_1(\alpha))\in\newH_\ell(F),
		\]
		its coefficient of \(X^{t+1}\) is zero for every \(\alpha\in\LL\). 
		
		We claim
		that this coefficient forces \(\beta_t\in\K\). Indeed, since \(t\) is a positive multiple of \(n\) and \(t<m\), we have
		\[
		t\le m-n.
		\]
		As \(n\ge3\), it follows that
		\[
		t+2\le m-n+2\le m-1.
		\]
		Thus both \(t+1\) and \(t+2\) belong to the chosen set of representatives
		\([-m,m-1]\). Moreover,
		\[
		\ell<t+1<m,
		\]
		so \(X^{t+1}\) lies outside the interval \([-\ell,\ell]\). Using the multiplication rule \(X^i a=\sigma^i(a)X^i\), the only direct
		contributions to the coefficient of \(X^{t+1}\) in
		\[
		\Delta_{\overline b}(Z_1(\alpha))
		=
		\overline b Z_1(\alpha)+Z_1(\alpha)\Theta(\overline b)
		\]
		come from
		\[
		\beta_tX^t\cdot \alpha X,
		\qquad
		\beta_{t+2}X^{t+2}\cdot\sigma^{-1}(\alpha)X^{-1},
		\]
		and from
		\[
		\alpha X\cdot\gamma_tX^t,
		\qquad
		\sigma^{-1}(\alpha)X^{-1}\cdot\gamma_{t+2}X^{t+2}.
		\]
		Their contributions are respectively
		\[
		\beta_t\sigma^t(\alpha),
		\qquad
		\beta_{t+2}\sigma^{t+2}(\sigma^{-1}(\alpha)),
		\]
		and
		\[
		\alpha\sigma(\gamma_t),
		\qquad
		\sigma^{-1}(\alpha)\sigma^{-1}(\gamma_{t+2}).
		\]
		Since \(t\equiv0\pmod n\), we have \(\sigma^t=\id_{\LL}\), and hence
		$
		\sigma^{t+2}(\sigma^{-1}(\alpha))=\sigma(\alpha)$.
		Therefore the coefficient of \(X^{t+1}\) is
		\[
		\bigl(\beta_t+\sigma(\gamma_t)\bigr)\alpha
		+
		\beta_{t+2}\sigma(\alpha)
		+
		\sigma^{-1}(\gamma_{t+2})\sigma^{-1}(\alpha).
		\]
		
		Now \(t+2>\ell\), and, since \(t\equiv0\pmod n\) and \(n\ge3\),
		\[
		t+2\not\equiv0\pmod n.
		\]
		Thus \(t+2\) is the index of a noncentral exponent outside the interval $[-\ell,\ell]$. By the previous step,
		\[
		\beta_{t+2}=0.
		\]
		Moreover, by \eqref{eq:outside-constant-relation},
		\[
		\gamma_{t+2}=-\beta_{t+2}=0.
		\]
		Hence the coefficient of \(X^{t+1}\) reduces to
		\[
		\bigl(\beta_t+\sigma(\gamma_t)\bigr)\alpha.
		\]
		It is zero for every \(\alpha\in\LL\). Taking \(\alpha=1\), we get
		\[
		\beta_t+\sigma(\gamma_t)=0.
		\]
		
		On the other hand, by applying again \eqref{eq:outside-constant-relation}, we have
		\[
		\gamma_t=-\beta_t.
		\]
		Consequently
		\[
		0=\beta_t+\sigma(\gamma_t)
		=
		\beta_t-\sigma(\beta_t).
		\]
		Thus
		\[
		\beta_t\in\Fix(\sigma)=\K.
		\]
		This proves that every positive central coefficient of \(\overline b\) having support in the set $[-m+1,m-1]\setminus [-\ell,\ell]$ belongs to \(\K\).
		
		We next handle the boundary coefficient \(\beta_{-m}\). We compare the
		coefficient of \(X^{-m+1}\) in
		\[
		\Delta_{\overline b}(Z_1(\alpha)).
		\]
		The exponent \(-m+1\) lies outside the interval, since \(\ell<m/2\). The relevant
		terms are
		\[
		\beta_{-m}X^{-m}\cdot \alpha X,
		\qquad
		\beta_{-m+2}X^{-m+2}\cdot\sigma^{-1}(\alpha)X^{-1},
		\]
		and
		\[
		\alpha X\cdot\gamma_{-m}X^{-m},
		\qquad
		\sigma^{-1}(\alpha)X^{-1}\cdot\gamma_{-m+2}X^{-m+2}.
		\]
		Their contributions to the coefficient of \(X^{-m+1}\) are
		\[
		\beta_{-m}\sigma^{-m}(\alpha),
		\qquad
		\beta_{-m+2}\sigma^{-m+2}(\sigma^{-1}(\alpha)),
		\]
		and
		\[
		\alpha\sigma(\gamma_{-m}),
		\qquad
		\sigma^{-1}(\alpha)\sigma^{-1}(\gamma_{-m+2}).
		\]
		Here possible boundary
		rewritings cannot contribute to this term since $n \geq 3$. Since \(m\equiv0\pmod n\), this gives 
		\[
		\bigl(\beta_{-m}+\sigma(\gamma_{-m})\bigr)\alpha
		+
		\beta_{-m+2}\sigma(\alpha)
		+
		\sigma^{-1}(\gamma_{-m+2})\sigma^{-1}(\alpha)
		=
		0
		\]
		for every \(\alpha\in\LL\).
		Since $n\ge3$, the three automorphisms
		\[
		\id_{\LL},\qquad \sigma,\qquad \sigma^{-1}
		\]
		are pairwise distinct and they are linearly independent over $\LL$ as maps
		$\LL\to\LL$. Consequently, the identity above, valid for every
		$\alpha\in\LL$, implies
		\[
			\beta_{-m+2}=0,
			\qquad
			\gamma_{-m+2}=0,
			\qquad
			\beta_{-m}+\sigma(\gamma_{-m})=0.
		\]
		By \eqref{eq:outside-constant-relation}, \(\gamma_{-m}=-\beta_{-m}\). Hence
		\[
		\beta_{-m}-\sigma(\beta_{-m})=0,
		\]
		and therefore
		\[
		\beta_{-m}\in\K.
		\]
		
		It remains to handle the negative central coefficients outside the interval $[-\ell,\ell]$. Arguing as before, we get 
		\[
		\beta_r=-\gamma_r\in\K,
		\]
		for every central exponent \(r\) such that
		\(-m<r<-\ell\).
		
		Thus every central coefficient of \(\overline b\) outside the interval $[-\ell,\ell]$ belongs to
		\(\K\). Consequently the outside central part of \(\overline b\) is an element of the
		center \(\E\). We may therefore write
		\[
		\overline b=\overline{c}+\lambda_0,
		\]
		where \(\lambda_0\in\E_F\) and
		\[
		\overline{c}=
		\sum_{i=-\ell}^{\ell}\beta_iX^i+RF(X^n).
		\]
		Since \(\overline{c}\) is supported in \([-\ell,\ell]\), and since
		\(\ell<m/2<m\), applying \(\Theta\) to \(\overline{c}\), we get
		\[
		\Theta(\overline{c})
		=
		\sum_{i=-\ell}^{\ell}\sigma^{-i}(\beta_i)X^{-i}+RF(X^n).
		\]
		Therefore \(\overline{c}+\Theta(\overline{c})\) is fixed by \(\Theta\) and supported in
		\([-\ell,\ell]\). This
		implies
		\[
		\overline{c}+\Theta(\overline{c})\in\newH_\ell(F).
		\]
		Since also
		\[
		\overline b+\Theta(\overline b)\in\newH_\ell(F),
		\]
		we get
		\[
		\overline{u}:=\lambda_0+\Theta(\lambda_0)
		=
		\bigl(\overline b+\Theta(\overline b)\bigr)
		-
		\bigl(\overline{c}+\Theta(\overline{c})\bigr)
		\in\newH_\ell(F).
		\]
		In particular, \(\overline{u}\) is a central element supported in the interval $[-\ell,\ell]$, and its
		central coefficients are in \(\K\).
		
		\noindent\textbf{Step 3: On the coefficients inside the interval \([-\ell,\ell]\).}
		
		We now eliminate the nonconstant part of \(\overline{c}\), modulo central terms. For
		\(1\le j\le\ell\) and \(\alpha\in\LL\), set
		\[
		Z_j(\alpha)
		:=
		\sigma^{-j}(\alpha)X^{-j}+\alpha X^j+RF(X^n)
		\in\newH_\ell(F).
		\]
		We proceed by descending induction on \(h=\ell,\ell-1,\ldots,1\). At the
		beginning of the step \(h\), we write
		\[
		\overline b=\overline{c}_h+\lambda_h,
		\]
		where \(\lambda_h\in\E\),
		\[
		\overline{u}_h:=\lambda_h+\Theta(\lambda_h)\in\newH_\ell(F),
		\]
		and \(\overline{c}_h\) is supported in \([-h,h]\) together with the still untreated
		coefficients at \(\pm h\). More precisely, all coefficients with absolute value
		strictly larger than \(h\) have already either vanished, if their exponent is
		not divisible by \(n\), or have been absorbed into \(\lambda_h\), if their exponent is
		divisible by \(n\). Put
		\[
		j:=\ell+1-h.
		\]
		Then \(1\le j\le\ell\). We compare the coefficient of \(X^{\ell+1}\) in
		\[
		\Delta_{\overline b}(Z_j(\alpha))
		=
		\Delta_{\overline{c}_h}(Z_j(\alpha))+\overline{u}_hZ_j(\alpha).
		\]
		Here \(\overline{c}_h\), \(\overline{u}_h\), and \(Z_j(\alpha)\) are all supported in exponents whose
		sums lie in \([-2\ell,2\ell]\). Since \(\ell<m/2\), this interval is contained
		in \([-m+1,m-1]\). Hence no reduction modulo \(F(X^n)\) is involved in this
		coefficient comparison.
		
		The only contribution from \(\overline{c}_h\) to the coefficient of \(X^{\ell+1}\) comes
		from the pair \(\beta_h,\beta_{-h}\). It is
		\[
		\beta_h\sigma^h(\alpha)
		+
		\alpha\sigma^{\ell+1}(\beta_{-h}).
		\]
		The central summand contributes through \(\overline{u}_hZ_j(\alpha)\). Such a contribution
		can occur only when \(h\equiv0\pmod n\). Let \(u_h\in\K\) denote the coefficient
		of \(X^h\) in \(\overline{u}_h\), with the convention \(u_h=0\) if
		\(h\not\equiv0\pmod n\). Since \(X^{\ell+1}\) lies outside the interval, we obtain
		\begin{equation}\label{eq:inner-coeff-comparison}
			\beta_h\sigma^h(\alpha)
			+
			\alpha\sigma^{\ell+1}(\beta_{-h})
			+
			u_h\alpha
			=
			0,
		\end{equation}
		for every $\alpha\in\LL$. Assume first that \(h\not\equiv0\pmod n\). Then \(u_h=0\). Taking
		\(\alpha=1\) in \eqref{eq:inner-coeff-comparison} gives
		\[
		\beta_h+\sigma^{\ell+1}(\beta_{-h})=0.
		\]
		Substituting this into \eqref{eq:inner-coeff-comparison}, we get
		\[
		\beta_h\bigl(\sigma^h(\alpha)-\alpha\bigr)=0,
		\]
		for every \(\alpha\in\LL\). Since \(\sigma^h\ne\id_{\LL}\), it follows that
		\[
		\beta_h=0,
		\qquad
		\beta_{-h}=0.
		\]
		
		Assume now that \(h\equiv0\pmod n\). Then \(\sigma^h=\id_{\LL}\), and
		\[
		\ell+1=h+j.
		\]
		Thus \eqref{eq:inner-coeff-comparison} gives
		\begin{equation}\label{eq:central-inner-first}
			\beta_h+\sigma^j(\beta_{-h})+u_h=0.
		\end{equation}
		Since \(n\ge3\), the condition \(h\equiv0\pmod n\) implies \(h\ge3\), and hence
		\(j+1\le\ell\). Applying the same comparison to \(Z_{j+1}(\alpha)\) and to the
		coefficient of \(X^{\ell+2}\), we obtain
		\begin{equation}\label{eq:central-inner-second}
			\beta_h+\sigma^{j+1}(\beta_{-h})+u_h=0.
		\end{equation}
		Subtracting \eqref{eq:central-inner-first} from
		\eqref{eq:central-inner-second} gives
		\[
		\sigma^{j+1}(\beta_{-h})=\sigma^j(\beta_{-h}),
		\]
		and therefore
		\[
		\sigma(\beta_{-h})=\beta_{-h}.
		\]
		Hence
		\[
		\beta_{-h}\in\K.
		\]
		Since \(u_h\in\K\), equation \eqref{eq:central-inner-first} also gives
		\[
		\beta_h\in\K.
		\]
		Therefore
		\[
		\beta_hX^h+\beta_{-h}X^{-h}\in\E.
		\]
		We absorb this central pair into the central summand \(\lambda_h\), and continue the
		induction. Iterating this argument for \(h=\ell,\ell-1,\ldots,1\), we obtain
		\[
		\overline b=\beta+\lambda
		\]
		with \(\beta\in\LL\) and \(\lambda\in\E\).
		
		\noindent\textbf{Step 4: Determination of the central summand.}
		
		It remains to determine which central elements \(\lambda\) may occur. Since
		$
		\beta\in\LL\subseteq\mathcal L(\newH_\ell(F))
		$
		and
		$
		\overline b\in\mathcal L(\newH_\ell(F)),
		$
		we have
		\[
		\lambda=\overline b-\beta\in\mathcal L(\newH_\ell(F)).
		\]
		Because \(\lambda\) is central, for every \(Z\in\newH_\ell(F)\),
		\[
		\Delta_{\lambda}(Z)
		=
		\lambda Z+Z\Theta(\lambda)
		=
		(\lambda+\Theta(\lambda))Z.
		\]
		Set
		\[
		\overline{u}:=\lambda+\Theta(\lambda)\in\E.
		\]
		Since \(1+RF(X^n)\in\newH_\ell(F)\), we have
		\[
		\overline{u}=\Delta_{\lambda}(1+RF(X^n))\in\newH_\ell(F).
		\]
		We claim that \(\overline{u}\in\K\).
		
		Suppose, to the contrary, that \(\overline{u}\notin\K\). Since
		\(\overline{u}\in\E\cap\newH_\ell(F)\), its support consists only of central
		exponents lying in \([-\ell,\ell]\). Let \(t=\xi n>0\) be the largest positive
		central exponent occurring in monomials of \(\overline{u}\), and let
		\(u_t\in\K^\times\) be its coefficient. Then \(t\le\ell\). Put
		\[
		j:=\ell+1-t.
		\]
		Then \(1\le j\le\ell\), and hence
		$
		Z_j(\alpha)\in\newH_\ell(F)$, for every $\alpha\in\LL$. The coefficient of \(X^{\ell+1}\) in
		$
		\overline{u}Z_j(\alpha)$
		is precisely
		$
		u_t\alpha.
		$
		Since
		\[
		\overline{u}Z_j(\alpha)=\Delta_{\lambda}(Z_j(\alpha))\in\newH_\ell(F),
		\]
		and \(X^{\ell+1}\) lies outside the interval $[-\ell,\ell]$, this coefficient must be zero for
		every \(\alpha\in\LL\). This forces \(u_t=0\), a contradiction. Therefore \(\overline{u}\) has no positive nonconstant central term. Since
		\(\overline{u}=\Theta(\overline{u})\), it has no negative nonconstant central term either. Hence
		\[
		\overline{u}\in\K.
		\]
		Thus
		$
		\lambda+\Theta(\lambda)\in\K,
		$
		that is,
		\[
		\lambda\in\mathcal B_F.
		\]
		Consequently
		\[
		\overline b=\beta+\lambda\in\LL+\mathcal B_F.
		\]
		This proves
		\[
		\mathcal L(\newH_\ell(F))\subseteq\LL+\mathcal B_F.
		\]
		Together with Proposition \ref{prop:easy-inclusion-lyapunov}, the theorem follows.
	\end{proof}

	\section{Equivalence}
	\label{sec:equivalence}

	\subsection{Main results}
	\label{subsec:main-results-equivalence}
	
	The construction in Section~\ref{sec:newconstruction} is parametrized by
	monic irreducible self-reciprocal polynomials.  Since many such polynomials
	exist in each prescribed even degree, a natural question is how many
	inequivalent members arise from this construction.
	
	Our main goal in this section is to answer this question.  We begin with the two main
	equivalence theorems---one for each parameter range---and their enumerative consequences for full Hermitian equivalence.
	
	We first fix the notation used throughout this section:
	\[
	\LL=\F_{q^n},
	\qquad
	\K=\F_q,
	\qquad
	R=\LL[X;\sigma],
	\]
	where \(\sigma\) generates \(\Gal(\LL/\K)\), and we assume that \(n>2\).
	Let \(s\ge2\) be even, and put
	\[
	\E:=\F_{q^s},
	\qquad
	\F:=\F_{q^{s/2}}.
	\]
	The nontrivial automorphism of the quadratic extension \(\E/\F\) is
	\[
	\varrho:\E\longrightarrow\E,
	\qquad
	x\longmapsto x^{q^{s/2}}.
	\]
	Let \(F(Y),G(Y)\in\K[Y]\) be monic irreducible self-reciprocal
	polynomials of degree \(s\).
	
	Choose roots \(u_F,u_G\in\E\) of \(F\) and \(G\), respectively.  Since
	\(u_H^\varrho=u_H^{-1}\) for \(H\in\{F,G\}\), the elements
	\[
	w_H:=u_H+u_H^{-1}
	\]
	belong to \(\F\).
	
	The next two theorems are the main results of this section about the linear Hermitian equivalence between $\newH_\ell(F)$ and $\newH_\ell(G)$.
	\begin{theorem}[Short-range equivalence]
		\label{thm:Hermitian-equivalence-polynomial-condition}
		Assume that
		\[
		n>2
		\qquad\text{and}\qquad
		1\le \ell<\min\{n,ns/4\}.
		\]
		If \(\newH_\ell(F)\) and \(\newH_\ell(G)\) are linearly Hermitian
		equivalent, then there exists \(\lambda\in\F_q^\times\) such that
		\[
		F(Y)=\lambda^sG(\lambda^{-1}Y).
		\]
	\end{theorem}
	
	\begin{theorem}[Large-range equivalence]
		\label{thm:Hermitian-equivalence-PGL2-large-ell}
		Assume that
		\[
		n>2,
		\qquad
		s>4,
		\qquad
		n<\ell<\frac{ns}{4}.
		\]
		If \(\newH_\ell(F)\) and \(\newH_\ell(G)\) are linearly Hermitian
		equivalent, then there exist \(a,b,c,d\in\K\), with \(ad-bc\ne0\), such that
		\[
		w_G=\frac{aw_F+b}{cw_F+d}.
		\]
	\end{theorem}
	Since \(w_F\notin\K\), the denominator \(cw_F+d\) is nonzero for every
	nonzero pair \((c,d)\in\K^2\).  Hence the fractional-linear expression in
	Theorem \ref{thm:Hermitian-equivalence-PGL2-large-ell} is well defined.
	
	For the counting statement, let
	\[
	\mathcal I^{\mathrm{sr}}_{q,s}
	:=
	\left\{
	H(Y)\in\F_q[Y]:
	H\text{ is monic, irreducible, self-reciprocal, and }\deg(H)=s
	\right\},
	\]
	and write
	\[
	N_q(s):=\bigl|\mathcal I^{\mathrm{sr}}_{q,s}\bigr|.
	\]
	For an integer \(d\) with \(1\le d\le n\), set
	\begin{equation}\label{eq:ell-d-main-results}
		\ell_d:=\frac{s(n-d+1)-2}{2}.
	\end{equation}
	Thus \(\newH_{\ell_d}(H)\) is the maximum additive Hermitian \(d\)-code
	constructed in Theorem~\ref{th:new-construction}.
	For the quadratic case, i.e.,~$s=2$, put
	\[
	L_2(q):=
	\begin{cases}
	\dfrac q2, & q\textup{ even},\\[2ex]
	\dfrac{q-1}{4}, & q\equiv1\pmod4,\\[2ex]
	\dfrac{q+1}{4}, & q\equiv3\pmod4.
	\end{cases}
	\]
	
	\begin{corollary}[Hermitian enumerative bounds]
		\label{cor:enumerative-Hermitian-combined}
		\label{cor:enumerative-s-equal-2-semilinear}
		\label{cor:enumerative-semilinear-s-geq-4-short}
		\label{cor:enumerative-semilinear-large-ell-PGL2}
		Write \(q=p^h\) and assume that \(n>2\).  For each even integer \(s\ge2\), let
		\[
			\Lambda_{s,d}:= \left\{
			\newH_{\ell_d}(H):H\in\mathcal I^{\mathrm{sr}}_{q,s}
			\right\}.
		\]
		The following lower bounds hold.
		\begin{enumerate}[label=\textup{(\roman*)},leftmargin=*]
			\item If \(s=2\) and \(n/2<d<n\), then $\Lambda_{s,d}$ contains at least 
			\[ 
			\left\lceil\frac{L_2(q)}{2h}\right\rceil
			\]
			pairwise Hermitian inequivalent maximum additive Hermitian \(d\)-codes in
			\(\mathcal H_n(\F_{q^2}/\F_q)\).
			
			\item If \(s\ge4\) is even and \(\ell_d<n\), then $\Lambda_{s,d}$
			contains at least
			\[
			\left\lceil
			\frac{N_q(s)}{sh\,\gcd(s,q-1)}
			\right\rceil
			\]
			pairwise Hermitian inequivalent maximum additive Hermitian \(d\)-codes in
			\(\mathcal H_n(\F_{q^s}/\F_{q^{s/2}})\).
			
			\item If \(s>4\) is even and
			\[
			n<\ell_d<\frac{ns}{4},
			\]
			then $\Lambda_{s,d}$
			contains at least
			\[
			\left\lceil
			\frac{N_q(s)}{sh\,q(q^2-1)}
			\right\rceil
			\]
			pairwise Hermitian inequivalent maximum additive Hermitian \(d\)-codes in
			\(\mathcal H_n(\F_{q^s}/\F_{q^{s/2}})\).
		\end{enumerate}
	\end{corollary}
	
	Since \(N_q(s)=\Theta(q^{s/2})\) as $q\rightarrow \infty$, the three parts have orders at least
	\(q/h\), \(q^{s/2}/h\), and \(q^{s/2-3}/h\), respectively, for fixed
	admissible \(s\).  The sharper linearly Hermitian bounds used to prove this
	corollary are stated later in Subsection~\ref{subsec:enumerative-consequences},
	where the relevant parameter-orbit arguments occur.
		
	\subsection{Proof strategy}
	For \(H\in\{F,G\}\), put
	\[
	R_H:=R/RH(X^n).
	\]
	For \(F\neq G\), the codes \(\newH_\ell(F)\subseteq R_F\) and
	\(\newH_\ell(G)\subseteq R_G\) lie in different quotient algebras, so a
	direct comparison is inconvenient.  To place them in a common setting,
	Theorem~\ref{thm:external-semilinear-realization} embeds each quotient
	\(R_H\), for \(H\in\{F,G\}\), into \(\End_{\E}(V)\).  More precisely, it
	provides an irreducible
	\(\sigma\)-semilinear operator \(T_H\) satisfying
	\(T_H^n=u_H\id_V\) and an \(\E\)-algebra isomorphism
	\(\Psi_H:R_H\to\End_\E(V)\) that sends the class of \(X\) to \(T_H\). 
	
	We write
	\[
	\mathfrak C_\ell(H):=\Psi_H(\newH_\ell(H))
	\qquad\text{and}\qquad
	\mathcal B:=\{\lambda\in\E:\lambda+\lambda^\varrho\in\K\}.
	\]
	The involutions induced by \(F\) and \(G\) on the common endomorphism
	algebra need not coincide.  Thus a linear Hermitian equivalence does
	not immediately have the form of a single Hermitian congruence inside
	\(\End_\E(V)\).  Instead, transporting the equivalence through
	separate identifications with the standard Hermitian matrix space
	produces an \(\E\)-algebra automorphism
	\[
	\eta\in\Aut_\E\bigl(\End_\E(V)\bigr)
	\]
	and an invertible element \(C\in\End_\E(V)^\times\) such that
	\[
	\mathfrak C_\ell(G)
	=
	\eta\bigl(\mathfrak C_\ell(F)\bigr)C.
	\]
	The Lyapunov invariant then gives
	\[
	\eta(\LL+\mathcal B)=\LL+\mathcal B.
	\]

	The proofs of both main equivalence theorems begin with the same chain of
	reductions:
	\[
	\begin{tikzcd}[row sep=large]
		\text{linear Hermitian equivalence of }
		\newH_\ell(F)\text{ and }\newH_\ell(G)
		\arrow[d, "\text{common semilinear realization}"]\\
		\text{linear Hermitian equivalence of }
		\mathfrak C_\ell(F)\text{ and }\mathfrak C_\ell(G)
		\arrow[d, "\text{Lyapunov invariance}"]\\
		Q(\LL+\mathcal B)Q^{-1}=\LL+\mathcal B
		\arrow[d, "\text{field normalization}"]\\
		Q\LL Q^{-1}=\LL.
	\end{tikzcd}
	\]
	The semilinear realization is established in
	Theorem~\ref{thm:external-semilinear-realization} and
	Proposition~\ref{prop:skew-vs-semilinear-linear}.  The Lyapunov computation
	gives \(\mathcal L(\mathfrak C_\ell(H))=\LL+\mathcal B\), while
	Theorem~\ref{thm:stabilizer-L-plus-U} shows that stabilizing
	\(\LL+\mathcal B\) already forces normalization of \(\LL\).
	
	After this common reduction, the proof divides according to the value of
	\(\ell\):
	\[
	\begin{tikzcd}[column sep=large,row sep=large]
		1\le\ell<\min\{n,ns/4\}
		\arrow[d, "\LL\text{-linear span}"]
		&
		n<\ell<ns/4
		\arrow[d, "\text{degree-zero component}"]\\
		\left\langle T_H^{-\ell},\ldots,1,\ldots,T_H^\ell\right\rangle_{\LL}
		\arrow[d, "\text{short power-space classification}"]
		&
		\left\langle1,w_H,\ldots,w_H^{\lfloor\ell/n\rfloor}\right\rangle_{\LL}
		\arrow[d, "\text{central interval determination}"]\\
		F(Y)=\lambda^sG(\lambda^{-1}Y)
		&
		w_G=\frac{aw_F+b}{cw_F+d}
	\end{tikzcd}
	\]
	In the short range, Lemma~\ref{lem:L-span-Hermitian-band-short} recovers the
	full interval of powers and
	Proposition~\ref{prop:short-power-equivalence-arbitrary-length} yields the
	scaling relation in
	Theorem~\ref{thm:Hermitian-equivalence-polynomial-condition}.  In the large
	range, Lemmas~\ref{lem:central-component-Hermitian-band} and~\ref{lem:central-interval-determination} recover the fractional-linear relation
	between \(w_F\) and \(w_G\), proving
	Theorem~\ref{thm:Hermitian-equivalence-PGL2-large-ell}.
	
	Finally, the linearly Hermitian enumeration results follow by bounding the
	size of the relevant parameter orbits: the quadratic scaling orbit is governed
	by \(c\mapsto-c\) on \(Y^2+cY+1\), the short-range scaling orbit has size at
	most \(\gcd(s,q-1)\), and a self-reciprocal
	\(\mathrm{PGL}_2(\F_q)\)-orbit has size at most \(q(q^2-1)\).  Passing from
	linear to full Hermitian equivalence introduces at most
	\(|\Aut(\F_{q^s})|=sh\) further identifications (and \(2h\) when \(s=2\)).
	This gives the three parts of
	Corollary~\ref{cor:enumerative-Hermitian-combined}.

	\subsection{A common semilinear realization}
	Although \(\newH_\ell(F)\subseteq R_F\) and
	\(\newH_\ell(G)\subseteq R_G\) lie in different central simple algebras, the
	following construction places both codes in a single endomorphism algebra.
	We first record the required semilinear terminology.
	
	Let \(V\) be a finite-dimensional
	\(\LL\)-vector space. For \(\alpha\in\LL\), we denote by
	\[
	m_\alpha:V\longrightarrow V,
	\qquad
	v\longmapsto \alpha v,
	\]
	the \(\LL\)-linear scalar multiplication by \(\alpha\).  A map
	$
	T:V\longrightarrow V
	$
	is called \textbf{\(\sigma\)-semilinear} if it is additive and
	\[
	T(\alpha v)=\sigma(\alpha)T(v)
	\qquad
	\text{for every }\alpha\in\LL,\ v\in V.
	\]
	In this paper, semilinear operators are always assumed to be bijective. If \(T\) is \(\sigma\)-semilinear, then \(T^i\) is
	\(\sigma^i\)-semilinear for every \(i\geq 0\).  Since \(\sigma\) has order
	\(n\), the operator \(T^n\) is \(\LL\)-linear.  Moreover,
	$
	T^nT=TT^n$
	and
	$
	T^n m_\alpha=m_\alpha T^n$,
	for every $\alpha\in\LL$.  Hence
	\(\K[T^n]\) is a central subalgebra of the \(\K\)-algebra generated by
	$
	m_\alpha$ and 
	$T$.
	This is the semilinear counterpart of the fact that \(X^n\) is central in
	the skew-polynomial ring \(\LL[X;\sigma]\).
	
	An \(\LL\)-subspace \(W\leq V\) is said to be \textbf{\(T\)-invariant} if
	$
	T(W)\subseteq W$.
	Since \(T\) is bijective, this is equivalent to \(T(W)=W\).  We say that
	\(T\) is \textbf{irreducible} if the only \(T\)-invariant \(\LL\)-subspaces of \(V\)
	are
	$
	\{0\}$ and $V$.
	
	Set
	\[
	r:=\gcd(n,s),
	\qquad
	n_0:=\frac nr,
	\qquad
	s_0:=\frac sr,
	\]
	and put
	\[
	\Omega:=\F_{q^{ns/r}}.
	\]
	We shall work on the fixed additive group
	\[
	V:=\Omega^r.
	\]
	We endow \(V\) with the following \(\LL\)-vector space structure:
	\[
	\alpha\cdot (x_0,\ldots,x_{r-1})
	=
	\bigl(
	\alpha x_0,\sigma(\alpha)x_1,\ldots,\sigma^{r-1}(\alpha)x_{r-1}
	\bigr),
	\qquad \alpha\in\LL.
	\]
	Then
	\[
	\dim_{\LL}V
	=
	r[\Omega:\LL]
	=
	rs_0
	=
	s.
	\]
	
	Now let \(H(Y)\neq Y\) be a monic irreducible polynomial in \(\K[Y]\) of
	degree \(s\), and choose a root
	$
	u_H\in\E=\F_{q^s}$
	of \(H\).  For \(u\in\E\subseteq\Omega\), it is clear that
	\[
	u\id_V:V\longrightarrow V \qquad (x_0,\ldots,x_{r-1})
	\longmapsto
	(ux_0,\ldots,ux_{r-1}).
	\]
	The important point is that the underlying \(\LL\)-vector space \(V\) is fixed
	independently of the defining polynomial \(H\).  Moreover,
	if \(T_H\) is chosen so that
	$
	T_H^n=u_H\id_V,$
	then
	\[
	\E_{T_H}:=\K[T_H^n]=\K[u_H\id_V]
	\]
	is a field of endomorphisms of \(V\), naturally isomorphic to
	\[
	\K(u_H)=\E.
	\]
	Thus \(V\) becomes a vector space over \(\E_{T_H}\), with
	$
	\dim_{\E_{T_H}}V=n.
	$
	
	The following result is proved in \cite[Theorem 3.6]{santonastaso2026autotopism}.
	\begin{theorem}
		\label{thm:external-semilinear-realization}
		Let \(H(Y)\neq Y\) be a monic irreducible polynomial in \(\K[Y]\) of degree
		\(s\), and let \(u_H\in\E\) be a root of \(H\).  Then there exists an
		irreducible \(\sigma\)-semilinear operator
		\[
		T_H:V\longrightarrow V
		\]
		on the \(s\)-dimensional \(\LL\)-vector space \(V=\Omega^r\) such that
		$
		T_H^n=u_H\id_V$. Moreover, the map
		\[
		\Psi_H:R_H\longrightarrow \End_{\E_{T_H}}(V)
		\]
		defined by
		\[
		\Psi_H\left(
		\sum_{i=0}^{ns-1}\alpha_iX^i+RH(X^n)
		\right)
		=
		\sum_{i=0}^{ns-1}m_{\alpha_i}T_H^i,
		\qquad
		\alpha_i\in\LL,
		\]
		is an \(\E\)-algebra isomorphism and an \(\LL\)-vector-space isomorphism.
	\end{theorem}
	
	In what follows, we shall write \(\alpha_iT_H^i\) for
	\(m_{\alpha_i}T_H^i\), by abuse of notation.
	
	For \(H\in\{F,G\}\), fix a root \(u_H\in\E\) of \(H\) and a map $\Psi_H$ as in \Cref{thm:external-semilinear-realization}.  Since \(H\) is
	self-reciprocal of even degree, one has
	$
	u_H^\varrho=u_H^{-1}$.
	
	Define
	\[
	\mathfrak C_\ell(H):=\Psi_H(\newH_\ell(H))
	\subseteq\End_\E(V).
	\]
	Then
	\begin{equation}\label{eq:semilinear-Hermitian-band}
		\mathfrak C_\ell(H)
		=
		\left\{
		\alpha_0+
		\sum_{i=1}^{\ell}
		\left(
		\sigma^{-i}(\alpha_i)T_H^{-i}
		+
		\alpha_iT_H^i
		\right)
		:
		\alpha_0,\ldots,\alpha_\ell\in\LL
		\right\}.
	\end{equation}
	
	Let
	\[
	\Theta_H:=\Psi_H\circ\Theta\circ\Psi_H^{-1}.
	\]
	Then \(\Theta_H\) is an anti-automorphism of \(\End_\E(V)\), and
	\begin{equation}\label{eq:thetaH-generators}
		\Theta_H(\alpha)=\alpha
		\quad\text{for every }\alpha\in\LL,
		\qquad
		\Theta_H(T_H)=T_H^{-1}.
	\end{equation}
	Moreover, \(\Theta_H\) restricts to \(\varrho\) on the center \(\E\id_V\).
	Indeed,
	\[
	\Theta_H(u_H\id_V)
	=
	\Theta_H(T_H^n)
	=
	T_H^{-n}
	=
	u_H^{-1}\id_V
	=
	u_H^\varrho\id_V.
	\]
	Thus \(\Theta_H\) is a unitary involution of \(\End_\E(V)\), and
	\(\mathfrak C_\ell(H)\subseteq\Fix(\Theta_H)\).

    We next show that linear
	Hermitian equivalence is preserved when one passes from the skew-polynomial
	model to the common semilinear model.  This justifies working from now on in
	\(\End_\E(V)\).
		Let
		\[
		\mathcal M_H:R_H\stackrel{\sim}{\longrightarrow}M_n(\E),
		\qquad H\in\{F,G\},
		\]
		be \(\E\)-algebra isomorphisms identifying the unitary involution on \(R_H\)
		with the standard Hermitian adjoint
		$
		A\mapsto A^*:=(A^\varrho)^\top
		$
		on \(M_n(\E)\). Likewise, let
		\[
		\mathcal N_H:\End_\E(V)\stackrel{\sim}{\longrightarrow}M_n(\E),
		\qquad H\in\{F,G\},
		\]
		be \(\E\)-algebra isomorphisms identifying the unitary involution
		\(\Theta_H\) on \(\End_\E(V)\) with the same adjoint \(A\mapsto A^*\).
		For convenience, we say that 
		\begin{itemize}
			\item two codes \(\newH_\ell(F)\) and \(\newH_\ell(G)\) are linearly Hermitian
			equivalent if $\mathcal{M}_F(\newH_\ell(F))$ and $\mathcal{M}_G(\newH_\ell(G))$ are linearly Hermitian
			equivalent;
			\item the codes \(\mathfrak C_\ell(F)\) and
			\(\mathfrak C_\ell(G)\) are linearly Hermitian equivalent if $\mathcal N_F(\mathfrak C_\ell(F))$ and $\mathcal N_G(\mathfrak C_\ell(G))$ are linearly Hermitian
			equivalent.
		\end{itemize}

	\begin{proposition}\label{prop:skew-vs-semilinear-linear}
		The codes \(\newH_\ell(F)\) and \(\newH_\ell(G)\) are linearly Hermitian
		equivalent if and only if the codes \(\mathfrak C_\ell(F)\) and
		\(\mathfrak C_\ell(G)\) are linearly Hermitian equivalent.
	\end{proposition}

	\begin{proof}
		For \(H\in\{F,G\}\), the map
		$
		\Psi_H:R_H\longrightarrow \End_{\E}(V)
		$
		is an \(\E\)-algebra isomorphism and maps
		$
		\newH_\ell(H)
		$
		onto
		$
		\mathfrak C_\ell(H).
		$ For each \(H\in\{F,G\}\), consider the commutative diagram
		\[
		\begin{tikzcd}
			R_H \arrow[r,"\Psi_H"] \arrow[d,"\mathcal M_H"'] &
			\End_\E(V) \arrow[d,"\mathcal N_H"] \\
			M_n(\E) \arrow[r,"\Phi_H"'] &
			M_n(\E)
		\end{tikzcd}
		\]
		where
		\[
		\Phi_H:=\mathcal N_H\circ\Psi_H\circ\mathcal M_H^{-1}.
		\]
		Then \(\Phi_H\) is an \(\E\)-algebra automorphism of \(M_n(\E)\).  By the
		Skolem--Noether theorem, there exists \(P_H\in\GL_n(\E)\) such that
		\[
		\Phi_H(A)=P_HAP_H^{-1}
		\qquad
		\text{for every }A\in M_n(\E).
		\]
		Since \(\mathcal M_H\), \(\Psi_H\), and \(\mathcal N_H\) all respect the
		corresponding unitary involutions, the automorphism \(\Phi_H\) commutes with
		\(A\mapsto A^*\).  Hence
		\[
		(P_HAP_H^{-1})^*=P_HA^* P_H^{-1}
		\qquad
		\text{for every }A\in M_n(\E).
		\]
		Equivalently,
		\[
		A\,P_H^* P_H=P_H^* P_H\,A
		\qquad
		\text{for every }A\in M_n(\E).
		\]
		Therefore
		\[
		P_H^* P_H=c_H I_n
		\]
		for some \(c_H\in\E^\times\).  Since \(P_H^* P_H\) is Hermitian, we
		have \(c_H^\varrho=c_H\), and so
		\[
		c_H\in\F^\times.
		\]
		Thus
		\[
		P_H^{-1}=c_H^{-1}P_H^*.
		\]
		Consequently,
		\begin{equation}\label{eq:semilinear-matrix-transfer-hermitian}
			\mathcal N_H(\mathfrak C_\ell(H))
			=
			c_H^{-1}P_H\,\mathcal M_H(\newH_\ell(H))\,P_H^*.
		\end{equation}
		
		Assume first that \(\newH_\ell(F)\) and \(\newH_\ell(G)\) are
		linearly Hermitian equivalent.  Then there exist
		\(U\in\GL_n(\E)\) and \(\lambda\in\F^\times\) such that
		\[
		\mathcal M_G(\newH_\ell(G))
		=
		\lambda\,U\,\mathcal M_F(\newH_\ell(F))\,U^*.
		\]
		Using \eqref{eq:semilinear-matrix-transfer-hermitian} for \(H=F\) and
		\(H=G\), we get
		\[
		\mathcal N_G(\mathfrak C_\ell(G))
		=
		\lambda\,\frac{c_F}{c_G}\,
		(P_GUP_F^{-1})\,
		\mathcal N_F(\mathfrak C_\ell(F))\,
		(P_GUP_F^{-1})^*.
		\]
		Since
		$
		\lambda\,\frac{c_F}{c_G}\in\F^\times,
		$
		this proves that \(\mathfrak C_\ell(F)\) and \(\mathfrak C_\ell(G)\) are
		linearly Hermitian equivalent. The converse is obtained by applying the same argument to
		\(\Psi_F^{-1}\) and \(\Psi_G^{-1}\).
	\end{proof}

    
	\subsection{The Lyapunov invariant and its stabilizer}
	We now translate the Lyapunov invariant to the common semilinear model.

	When the Lyapunov invariant is applied to
	\(\mathfrak C_\ell(H)\), it is computed with respect
	to the involution \(\Theta_H\).  Thus, temporarily writing the
	involution as a subscript,
	\[
	\mathcal L_{\Theta_H}\bigl(\mathfrak C_\ell(H)\bigr)
	:=
	\left\{
	B\in\End_\E(V):
	BA+A\Theta_H(B)\in\mathfrak C_\ell(H)
	\text{ for every }A\in\mathfrak C_\ell(H)
	\right\}.
	\]
	We shall continue to omit the subscript whenever the involution is
	clear from the code under consideration.
    
	Recall that
	\[
	\mathcal B
	=
	\{\lambda\in\E:\lambda+\lambda^\varrho\in\K\}.
	\]
	
	Throughout, we identify \(\E\) with the scalar subfield
	$
	\E\id_V\subseteq\End_\E(V),$
	and we identify \(\LL\) with the copy
	$
	m_\LL:=\{m_\alpha:\alpha\in\LL\}\subseteq\End_\E(V)
	$
	coming from the fixed \(\LL\)-vector space structure on \(V\).  By abuse of
	notation, we simply write \(\LL\) for \(m_\LL\).

	\begin{proposition}\label{prop:lyapunov-common-model}
		Assume $n>2$ and \(1\le \ell<ns/4\).  For every \(H\in\{F,G\}\), one has
		\[
		\mathcal L(\mathfrak C_\ell(H))
		=
		\LL+\mathcal B.
		\]
	\end{proposition}
	
	\begin{proof}
		The assertion follows by applying \(\Psi_H\) to
		\Cref{thm:lyapunov-invariant-computation}. Under \(\Psi_H\), the coefficient
		field \(\LL\subseteq R_H\) is identified with the copy
		\(\{m_\alpha:\alpha\in\LL\}\subseteq\End_\E(V)\).
		Moreover, the space
		\[
		\mathcal B_H
		=
		\{\lambda \in \E_F: \lambda+\Theta(\lambda)\in\K\}
		\]
		is identified with $\mathcal B$,
		because \(\Theta_H\) restricts to \(\varrho\) on the scalar field
		\(\E\id_V\).
	\end{proof}

	We next determine the stabilizer of \(\LL+\mathcal B\) in $\GL_{\E}(V)$.  We start with a preliminary lemma.
	
	\begin{lemma}\label{lem:EL-unique-expansion}
		The map
		\[
		\E\otimes_\K\LL\longrightarrow\End_\E(V),
		\qquad
		e\otimes \alpha \longmapsto e\alpha,
		\]
		is injective.  Also,
		$
		\E\cap\LL=\K$.
	\end{lemma}
	
	\begin{proof}
		Fix \(H\in\{F,G\}\), let \(u_H\in\E\) be a root of \(H\), and write
		\[
		z_H:=X^n+RH(X^n)\in Z(R_H).
		\]
		Via the map
		\[
		\Psi_H:R_H\longrightarrow \End_\E(V),
		\]
		the center \(Z(R_H)\cong\K[X^n]/(H(X^n))\) is identified with the scalar field \(\E\id_V\),
		via
		\[
		z_H\longmapsto u_H\id_V,
		\]
		whereas the coefficient field \(\LL\subseteq R_H\) is mapped to the subfield
		$
		m_\LL:=\{m_\alpha:\alpha\in\LL\}\subseteq\End_\E(V)
		$
		defined by the fixed \(\LL\)-vector space structure on \(V\). Since \(H\) is irreducible of degree \(s\), the elements
		\[
		1,u_H,\ldots,u_H^{s-1}
		\]
		form a \(\K\)-basis of \(\E\).  Hence every element of
		\(\E\otimes_\K\LL\) can be written uniquely as
		\[
		\sum_{j=0}^{s-1}u_H^j\otimes \alpha_j,
		\qquad
		\alpha_j\in\LL.
		\]
		Its image under the multiplication map is
		\[
		\sum_{j=0}^{s-1}u_H^j\alpha_j
		\in \End_\E(V).
		\]
		Therefore it is enough to prove that every relation
		\[
		\sum_{j=0}^{s-1}u_H^j\alpha_j=0,
		\qquad
		\alpha_j\in\LL,
		\]
		inside \(\End_\E(V)\), is trivial. Pulling this relation back through \(\Psi_H\), we obtain
		\[
		\sum_{j=0}^{s-1}z_H^j\alpha_j=0
		\qquad
		\text{in }R_H.
		\]
		Equivalently,
		\[
		\sum_{j=0}^{s-1}\alpha_jX^{nj}\in R\,H(X^n).
		\]
		Since,
		$
		\deg\left(\sum_{j=0}^{s-1}\alpha_jX^{nj}\right)<ns,
		$
		we get
		$
		\alpha_0=\alpha_1=\cdots=\alpha_{s-1}=0$.
		This proves that the multiplication map
		\[
		\E\otimes_\K\LL\longrightarrow\End_\E(V),
		\qquad
		e\otimes\alpha\longmapsto e\alpha,
		\]
		is injective. In particular,
		\[
		\E\cap\LL=\K. \qedhere
		\]
	\end{proof}

	\begin{lemma}\label{lem:subfields-inside-L-plus-E}
		Assume \(n>2\). Let \(W\subseteq \End_\E(V)\) be a unital subfield such that
		\[
		\K\subseteq W,\qquad
		\dim_\K W=n,\qquad
		W\cap \E=\K,
		\]
		and
		\[
		W\subseteq \LL+\E.
		\]
		Then
		\[
		W=\LL.
		\]
	\end{lemma}
	
	\begin{proof}
		Choose a \(\K\)-linear complement \(\E_0\) of \(\K\) in \(\E\), so that
		\[
		\E=\K\oplus \E_0.
		\]
		By \Cref{lem:EL-unique-expansion}, we have
		$
		\LL\cap \E=\K$. Hence
		\[
		\LL\cap \E_0=0.
		\]
		Since \(\K\subseteq \LL\), we obtain the direct sum decomposition
		\[
		\LL+\E=\LL+\K+\E_0=\LL\oplus \E_0
		\]
		as \(\K\)-vector spaces. Let
		\[
		\pi:\LL\oplus \E_0\longrightarrow \LL
		\]
		be the projection onto the first component.
		
		We claim that the restriction
		\[
		\pi|_W:W\longrightarrow \LL
		\]
		is injective. Indeed, if \(w\in W\) and \(\pi(w)=0\), then
		$
		w\in W\cap \E_0$.
		Since \(\E_0\subseteq\E\), we get
		\[
		w\in W\cap\E=\K.
		\]
		But \(w\in\E_0\), and \(\E_0\cap\K=0\). Therefore \(w=0\). Thus
		\(\pi|_W\) is injective. Since
		\[
		\dim_\K W=n=\dim_\K\LL,
		\]
		it follows that \(\pi|_W\) is an isomorphism.
		
		Consequently, there exists a unique \(\K\)-linear map
		\[
		f:\LL\longrightarrow \E_0
		\]
		such that
		\[
		W=\{a+f(a):a\in\LL\}.
		\]
		Since \(1\in W\), we have \(f(1)=0\).
		
		We now prove that \(f=0\). Choose a \(\K\)-basis
		\[
		e_0=1,e_1,\ldots,e_{s-1}
		\]
		of \(\E\) such that \(e_1,\ldots,e_{s-1}\) is a \(\K\)-basis of \(\E_0\). For every
		\(a\in\LL\), write
		\[
		f(a)=\sum_{i=1}^{s-1}\varphi_i(a)e_i,
		\]
		where \(\varphi_i:\LL\to\K\) are \(\K\)-linear maps.
		
		Let \(a,b\in\LL\). Since \(W\) is closed under multiplication, we have
		\[
		(a+f(a))(b+f(b))\in W\subseteq \LL+\E.
		\]
		Expanding the product, and using that \(ab\in\LL\) and \(f(a)f(b)\in\E\), we obtain
		\[
		f(a)b+f(b)a\in \LL+\E.
		\]
		Equivalently,
		\[
		\sum_{i=1}^{s-1}
		e_i\bigl(\varphi_i(a)b+\varphi_i(b)a\bigr)
		\in \LL+\E.
		\]
		By Lemma~\ref{lem:EL-unique-expansion}, applied with respect to the basis
		\(e_0=1,e_1,\ldots,e_{s-1}\), this implies that
		\[
		\varphi_i(a)b+\varphi_i(b)a\in\K
		\]
		for every \(i=1,\ldots,s-1\) and every \(a,b\in\LL\).
		
		We show that each \(\varphi_i\) is zero. Suppose, by contradiction, that
		\(\varphi_i\neq0\) for some \(i\). Since \(\varphi_i:\LL\to\K\) is a nonzero
		\(\K\)-linear map and \(\dim_\K\LL=n\), we have
		\[
		\dim_\K\ker(\varphi_i)=n-1.
		\]
		Because \(n>2\), this dimension is greater than \(1\). Hence
		\(\ker(\varphi_i)\) is not contained in the one-dimensional \(\K\)-subspace
		\(\K\subseteq\LL\). Choose
		\[
		b\in\ker(\varphi_i)\setminus\K.
		\]
		Since \(\varphi_i\neq0\), choose \(\alpha\in\LL\) such that
		\[
		\varphi_i(\alpha)\neq0.
		\]
		Then
		\[
		\varphi_i(\alpha)b+\varphi_i(b)\alpha\in\K
		\]
		and \(\varphi_i(b)=0\), so
		\[
		\varphi_i(\alpha)b\in\K.
		\]
		Since \(\varphi_i(\alpha)\in\K^\times\), this forces \(b\in\K\), a contradiction.
		
		Therefore \(\varphi_i=0\) for every \(i=1,\ldots,s-1\). Hence \(f=0\), and so
		\[
		W=\{\alpha:\alpha\in\LL\}=\LL.
		\]
		This proves the claim.
	\end{proof}
	
	\begin{theorem}\label{thm:stabilizer-L-plus-U}
		Assume \(n>2\). Let \(\mathcal U\subseteq \E\) be a \(\K\)-subspace with
		\(\K\subseteq \mathcal U\). Then
		\[
		\Stab_{\GL_\E(V)}(\LL+\mathcal U)
		=
		\Stab_{\GL_\E(V)}(\LL),
		\]
		where both stabilizers are taken with respect to the conjugation action of
		\(\GL_\E(V)\) on \(\End_\E(V)\).
	\end{theorem}
	
	\begin{proof}
		We need to prove that, for every
		\(Q\in\GL_\E(V)\),
		\[
		Q(\LL+\mathcal U)Q^{-1}=\LL+\mathcal U
		\]
		if and only if
		\[
		Q\LL Q^{-1}=\LL.
		\]
		
		Assume first that
		$
		Q\LL Q^{-1}=\LL$.
		Since \(\mathcal U\subseteq\E\id_V\), and since \(Q\) is \(\E\)-linear, one has
		$
		Q\mathcal UQ^{-1}=\mathcal U$.
		Therefore
		\[
		Q(\LL+\mathcal U)Q^{-1}
		=
		Q\LL Q^{-1}+Q\mathcal UQ^{-1}
		=
		\LL+\mathcal U.
		\]
		Thus
		\[
		\Stab_{\GL_\E(V)}(\LL)
		\subseteq
		\Stab_{\GL_\E(V)}(\LL+\mathcal U).
		\]
		
		Conversely, suppose that
		$
		Q(\LL+\mathcal U)Q^{-1}=\LL+\mathcal U.
		$
		Put
		$
		W:=Q\LL Q^{-1}$.
		Then
		\[
		W\subseteq \LL+\mathcal U\subseteq \LL+\E.
		\]
		Moreover, \(W\) is a unital subfield of \(A=\End_\E(V)\), contains \(\K\), and
		$
		\dim_\K W=n$.
		
		We prove that
		$
		W\cap\E=\K$.
		Let \(w\in W\cap\E\). Then
		$
		w=Q\alpha Q^{-1}$
		for some \(\alpha\in\LL\). Since \(w\in\E\id_V\), and since
		\(Q\in\GL_\E(V)\), we have
		$
		Q^{-1}wQ=w$.
		Hence
		\[
		\alpha=Q^{-1}wQ=w.
		\]
		Thus
		$
		w=\alpha\in\LL\cap\E=\K$
		by \Cref{lem:EL-unique-expansion}. Therefore
		$
		W\cap\E\subseteq\K$.
		The reverse inclusion is clear, since \(\K\subseteq W\) and
		\(\K\subseteq\E\). Hence
		\[
		W\cap\E=\K.
		\]
		
		By \Cref{lem:subfields-inside-L-plus-E}, we obtain
		$
		W=\LL$.
		Therefore
		\[
		Q\LL Q^{-1}=\LL.
		\]
		This proves the reverse inclusion, and hence the theorem.
	\end{proof}
	
	We can now derive the precise form of a linear Hermitian equivalence
	in the common endomorphism algebra.  Notice that the following
	argument does not require the involutions \(\Theta_F\) and
	\(\Theta_G\) to coincide.
	
	\begin{lemma}
	\label{lem:common-equivalence-reduction}
	Assume that \(n>2\) and
	$
	1\leq\ell<\frac{ns}{4}$.
	Suppose that \(\newH_\ell(F)\) and \(\newH_\ell(G)\) are linearly
	Hermitian equivalent.  Then there exist
	\[
	Q\in\GL_\E(V)
	\qquad\text{and}\qquad
	C\in\End_\E(V)^\times
	\]
	such that
	\[
	\mathfrak C_\ell(G)
	=
	Q\mathfrak C_\ell(F)Q^{-1}C,
	\]
	\[
	Q\LL Q^{-1}=\LL,
	\]
	and
	\[
	C\in \Cen_{\End_\E(V)}(\LL)^\times.
	\]
	\end{lemma}
	
	\begin{proof}
	For \(H\in\{F,G\}\), choose an \(\E\)-algebra isomorphism
	\[
	\mathcal N_H:
	\End_\E(V)
	\stackrel{\sim}{\longrightarrow}
	M_n(\E)
	\]
	identifying the involution \(\Theta_H\) with the standard Hermitian
	adjoint \(A\mapsto A^*\). By Proposition~\ref{prop:skew-vs-semilinear-linear}, the codes
	\(\mathfrak C_\ell(F)\) and \(\mathfrak C_\ell(G)\) are linearly
	Hermitian equivalent.  Hence there exist
	$
	S\in\GL_n(\E)$ and $\nu\in\F^\times $
	such that
	\[
	\mathcal N_G\bigl(\mathfrak C_\ell(G)\bigr)
	=
	\nu S\,
	\mathcal N_F\bigl(\mathfrak C_\ell(F)\bigr)
	S^*.
	\]
	
	Define
	\[
	\eta
	:=
	\mathcal N_G^{-1}
	\circ
	\tau_S
	\circ
	\mathcal N_F,
	\]
	where
	$
	\tau_S(A):=SAS^{-1},$ for every $A\in M_n(\E)$, and put
	\[
	C
	:=
	\mathcal N_G^{-1}(\nu SS^*).
	\]
	Then \(\eta\) is an \(\E\)-algebra automorphism of
	\(\End_\E(V)\), the element \(C\) is invertible, and
	\[
	\mathfrak C_\ell(G)
	=
	\eta\bigl(\mathfrak C_\ell(F)\bigr)C.
	\]
	We next transport the Lyapunov invariant.  Since
	\(\mathcal N_H\) is an isomorphism of algebras with involution, one
	has
	\[
	\mathcal N_H
	\left(
	\mathcal L_{\Theta_H}
	\bigl(\mathfrak C_\ell(H)\bigr)
	\right)
	=
	\mathcal L
	\left(
	\mathcal N_H\bigl(\mathfrak C_\ell(H)\bigr)
	\right).
	\]
	Therefore, by Proposition~\ref{prop:lyapunov-invariant-matrix},
	\[
	\begin{aligned}
	\mathcal N_G
	\left(
	\mathcal L_{\Theta_G}
	\bigl(\mathfrak C_\ell(G)\bigr)
	\right)
	&=
	S\,
	\mathcal N_F
	\left(
	\mathcal L_{\Theta_F}
	\bigl(\mathfrak C_\ell(F)\bigr)
	\right)
	S^{-1}.
	\end{aligned}
	\]
	Equivalently,
	\[
	\mathcal L_{\Theta_G}
	\bigl(\mathfrak C_\ell(G)\bigr)
	=
	\eta
	\left(
	\mathcal L_{\Theta_F}
	\bigl(\mathfrak C_\ell(F)\bigr)
	\right).
	\]
	By Proposition~\ref{prop:lyapunov-common-model},
	\[
	\mathcal L_{\Theta_F}
	\bigl(\mathfrak C_\ell(F)\bigr)
	=
	\LL+\mathcal B
	=
	\mathcal L_{\Theta_G}
	\bigl(\mathfrak C_\ell(G)\bigr).
	\]
	Consequently,
	\[
	\eta(\LL+\mathcal B)=\LL+\mathcal B.
	\]
	Since \(\End_\E(V)\cong M_n(\E)\), the Skolem--Noether theorem gives
	an element \(Q\in\GL_\E(V)\) such that
	\[
	\eta=\tau_Q.
	\]
	Thus
	\[
	Q(\LL+\mathcal B)Q^{-1}
	=
	\LL+\mathcal B.
	\]
	By Theorem~\ref{thm:stabilizer-L-plus-U}, it follows that
	\[
	Q\LL Q^{-1}=\LL.
	\]
	Moreover,
	\[
	\mathfrak C_\ell(G)
	=
	Q\mathfrak C_\ell(F)Q^{-1}C.
	\]
	It remains to prove that \(C\) centralizes \(\LL\).  Since
	\(1\in\mathfrak C_\ell(F)\), the equality above gives
	\[
	C\in\mathfrak C_\ell(G).
	\]
	Hence
	\[
	\Theta_G(C)=C.
	\]
	Furthermore, since
	$
	\LL\subseteq\mathfrak C_\ell(F)$ and $\eta(\LL)=\LL,$
	we have
	$
	\alpha C\in\mathfrak C_\ell(G)$, for every $\alpha\in\LL$.
	Every element of \(\mathfrak C_\ell(G)\) is fixed by \(\Theta_G\).
	Using also that \(\Theta_G\) fixes \(\LL\) pointwise, we obtain
	\[
	\alpha C
	=
	\Theta_G(\alpha C)
	=
	\Theta_G(C)\Theta_G(\alpha)
	=
	C\alpha.
	\]
	Thus \(C\) centralizes \(\LL\), and therefore
	$
	C\in \Cen_{\End_\E(V)}(\LL)^\times$.
	\end{proof}
	
	\subsection{The range $1 \leq \ell <n$}
	
	\begin{lemma}\label{lem:L-span-Hermitian-band-short}
		Assume \(1\le \ell<n\). For \(H\in\{F,G\}\), set
		\[
		\mathfrak P_\ell(H):=
		\langle \mathfrak C_\ell(H)\rangle_{\LL}.
		\]
		Then
		\[
		\mathfrak P_\ell(H)
		=
		\left\langle
		T_H^{-\ell},T_H^{-\ell+1},\ldots,T_H^{-1},1,
		T_H,\ldots,T_H^\ell
		\right\rangle_{\LL}.
		\]
	\end{lemma}
	
	\begin{proof}
		From \eqref{eq:semilinear-Hermitian-band}, we have
		\[
		\mathfrak C_\ell(H)
		=
		\left\{
		\alpha_0+
		\sum_{i=1}^{\ell}
		\left(
		\sigma^{-i}(\alpha_i)T_H^{-i}
		+
		\alpha_iT_H^i
		\right)
		:
		\alpha_0,\ldots,\alpha_\ell\in\LL
		\right\}.
		\]
		Hence
		\[
		\mathfrak P_\ell(H)
		\subseteq
		\left\langle
		T_H^{-\ell},T_H^{-\ell+1},\ldots,T_H^{-1},1,
		T_H,\ldots,T_H^\ell
		\right\rangle_{\LL}.
		\]
		
		We prove the reverse inclusion. Clearly \(1\in\mathfrak C_\ell(H)\), and
		therefore
		\[
		\LL\subseteq\mathfrak P_\ell(H).
		\]
		Let \(i\in\{1,\ldots,\ell\}\). Since \(\ell<n\), we have \(i<n\), and hence
		\[
		\sigma^i\neq \id_{\LL}.
		\]
		Thus there exists \(\eta_i\in\LL\) such that
		\[
		\sigma^{-i}(\eta_i)\neq \eta_i.
		\]
		The two elements
		\[
		T_H^{-i}+T_H^i
		\]
		and
		\[
		\sigma^{-i}(\eta_i)T_H^{-i}+\eta_iT_H^i
		\]
		belong to \(\mathfrak C_\ell(H)\subseteq\mathfrak P_\ell(H)\). Since
		\[
		\det
		\begin{pmatrix}
			1 & 1\\
			\sigma^{-i}(\eta_i) & \eta_i
		\end{pmatrix}
		=
		\eta_i-\sigma^{-i}(\eta_i)
		\neq0,
		\]
		their \(\LL\)-span contains both \(T_H^{-i}\) and \(T_H^i\). Therefore
		\[
		T_H^{-i},T_H^i\in\mathfrak P_\ell(H)
		\qquad
		\text{for every }i=1,\ldots,\ell.
		\]
		This proves
		\[
		\mathfrak P_\ell(H)
		=
		\left\langle
		T_H^{-\ell},T_H^{-\ell+1},\ldots,T_H^{-1},1,
		T_H,\ldots,T_H^\ell
		\right\rangle_{\LL}. \qedhere
		\]
	\end{proof}
	
	Note that \(T_H^\ell\) is invertible, and right multiplication by \(T_H^\ell\)
	gives
	\[
	\mathfrak S_{2\ell}(H):=\mathfrak P_\ell(H)T_H^\ell=
	\left\langle
	1,T_H,\ldots,T_H^{2\ell}
	\right\rangle_{\LL}.
	\]

	In \cite{santonastaso2026autotopism}, linear equivalence was studied for
	\(\LL\)-linear power spaces of the form
	\[
	\left\langle1,T_H,\ldots,T_H^{sk-1}\right\rangle_{\LL}.
	\]
	Although the result there is stated with last exponent \(sk-1\), its proof uses only that the exponents form an interval and that the largest exponent is	strictly smaller than \(ns/2\). The same argument therefore gives the	following more flexible statement.
		
	\begin{proposition}\label{prop:short-power-equivalence-arbitrary-length}
		Assume \(n>2\), and let \(j\) be an integer such that
		\[
		1\le j<\frac{ns}{2}.
		\]
		For \(H\in\{F,G\}\), set
		\[
		\mathfrak S_j(H):=
		\left\langle
		1,T_H,\ldots,T_H^j
		\right\rangle_{\LL}.
		\]
		Suppose that there exist \(M,N\in\GL_\E(V)\) such that
		\[
		\mathfrak S_m(G)=M\,\mathfrak S_m(F)\,N
		\]
		and
		\[
		M\LL M^{-1}=\LL.
		\]
		Then there exists \(\lambda\in\F_q^\times\) such that
		\[
		F(Y)=\lambda^sG(\lambda^{-1}Y).
		\]
	\end{proposition}

	We are now ready to prove the short-range main equivalence theorem stated at
	the beginning of this section.
	
	\begin{proof}[Proof of Theorem~\ref{thm:Hermitian-equivalence-polynomial-condition}]
		Assume that \(\newH_\ell(F)\) and \(\newH_\ell(G)\) are linearly
		Hermitian equivalent. By Lemma~\ref{lem:common-equivalence-reduction}, there exist
	\[
	Q\in\GL_\E(V)
	\qquad\text{and}\qquad
	C\in \Cen_{\End_\E(V)}(\LL)^\times
	\]
	such that
	\[
	\mathfrak C_\ell(G)
	=
	Q\mathfrak C_\ell(F)Q^{-1}C
	\]
	and
	\[
	Q\LL Q^{-1}=\LL.
	\]
		
	We now pass from the Hermitian codes to their \(\LL\)-linear spans. For \(H\in\{F,G\}\), set
	\[
		\mathfrak P_\ell(H):=
		\langle\mathfrak C_\ell(H)\rangle_{\LL}.
	\]
	Since \(Q\) normalizes \(\LL\), taking the \(\LL\)-linear span in the
	equivalence above gives
	\[
	\mathfrak P_\ell(G)
	=
	Q\mathfrak P_\ell(F)Q^{-1}C.
	\]
	
	By Lemma~\ref{lem:L-span-Hermitian-band-short}, since \(\ell<n\),
	\[
	\mathfrak P_\ell(H)
	=
	\left\langle
	T_H^{-\ell},\ldots,T_H^{-1},1,T_H,\ldots,T_H^\ell
	\right\rangle_{\LL}.
	\]
	Right multiplication by \(T_H^\ell\) gives
	\[
	\mathfrak S_{2\ell}(H)
	=
	\mathfrak P_\ell(H)T_H^\ell
	=
	\left\langle
	1,T_H,\ldots,T_H^{2\ell}
	\right\rangle_{\LL}.
	\]
	Consequently,
	\[
	\begin{aligned}
	\mathfrak S_{2\ell}(G)
	&=
	\mathfrak P_\ell(G)T_G^\ell\\
	&=
	Q\mathfrak P_\ell(F)Q^{-1}CT_G^\ell\\
	&=
	Q\mathfrak S_{2\ell}(F)
	\left(
	T_F^{-\ell}Q^{-1}CT_G^\ell
	\right).
	\end{aligned}
	\]
	Thus Proposition~\ref{prop:short-power-equivalence-arbitrary-length}
	applies with
	\[
	M:=Q
	\qquad\text{and}\qquad
	N:=T_F^{-\ell}Q^{-1}CT_G^\ell,
	\]
	because
	\[
	M\LL M^{-1}=Q\LL Q^{-1}=\LL\]
	and
	$	2\ell<ns/2$.
	\end{proof}
		
	\subsection{The range \(n<\ell<ns/4\)}
	\label{subsec:large-ell-PGL2}
	We now consider the range
	\[
		n<\ell<\frac{ns}{4}.
	\]
	This range is nonempty only when \(s>4\).  The short-range argument from
	Theorem~\ref{thm:Hermitian-equivalence-polynomial-condition} no longer
	applies, because the \(\LL\)-span of the Hermitian code need not contain every
	power of \(T_H\) in the interval \([-\ell,\ell]\); the obstruction occurs
	precisely at exponents divisible by \(n\).  Nevertheless, the central
	elements in this interval still provide a useful invariant.  They lead to a
	classification through the natural \(\mathrm{PGL}_2(\K)\)-action on
	\[
		u_H+u_H^{-1}.
	\]
	
	For \(H\in\{F,G\}\), recall that \(u_H\in\E=\F_{q^s}\) is a root of \(H\).
	Since \(H\) is self-reciprocal of even degree, we have
	\[
	u_H^\varrho=u_H^{-1},
	\]
	where \(\varrho:x\mapsto x^{q^{s/2}}\) is the nontrivial automorphism of
	\(\E/\F\).  We put
	\[
	w_H:=u_H+u_H^{-1}.
	\]
	Then \(w_H\in\F=\F_{q^{s/2}}\).
	
	\begin{lemma}\label{lem:degree-of-wH}
		Let \(H\in\K[Y]\) be monic, irreducible, and self-reciprocal of even
		degree \(s>2\), and let \(u_H\) be a root of \(H\).  Then
		\[
		[\K(w_H):\K]=\frac{s}{2}.
		\]
	\end{lemma}
	
	\begin{proof}
		We know that \(w_H\in\E^\varrho=\F_{q^{s/2}}\), and hence
		\[
		[\K(w_H):\K]\le \frac{s}{2}.
		\]
		On the other hand, \(u_H\) satisfies
		\[
		Y^2-w_HY+1=0
		\]
		over \(\K(w_H)\).  Therefore
		\[
		[\K(u_H):\K(w_H)]\le 2.
		\]
		Since \(H\) has degree \(s\), we have \([\K(u_H):\K]=s\).  Hence
		\[
		s=[\K(u_H):\K]
		=
		[\K(u_H):\K(w_H)]\,[\K(w_H):\K]
		\le
		2[\K(w_H):\K].
		\]
		It follows that
		\[
		[\K(w_H):\K]\ge \frac{s}{2}.
		\]
		Therefore
		\[
		[\K(w_H):\K]=\frac{s}{2}. \qedhere
		\]
	\end{proof}
	
	We work in the commutative \(\LL\)-algebra
	\[
	\mathfrak{A}:=\F\otimes_{\K}\LL.
	\]
	We identify \(\F\) with the subalgebra \(\F\otimes1\) and \(\LL\) with the
	subalgebra \(1\otimes\LL\).  Thus \(\mathfrak{A}\) is naturally an \(\LL\)-vector
	space.  For \(z\in\F\), we write simply \(z\) instead of \(z\otimes1\).
	With this convention,
	\[
	\left\langle1,w_H,\ldots,w_H^t\right\rangle_{\LL}
	\]
	denotes the \(\LL\)-subspace of \(\mathfrak{A}\) generated by
	\(1,w_H,\ldots,w_H^t\), where \(t\) is a nonnegative integer.
	\begin{lemma}
		\label{lem:central-interval-determination}
		Let \(\ell\) be an integer such that
		\[
		n<\ell<\frac{ns}{4},
		\]
		and set
		\[
		a_0:=\left\lfloor\frac{\ell}{n}\right\rfloor.
		\]
		Assume that
		\[
		\left\langle
		1,w_G,\ldots,w_G^{a_0}
		\right\rangle_{\LL}
		=
		\left\langle
		1,w_F,\ldots,w_F^{a_0}
		\right\rangle_{\LL}C
		\]
		for some invertible element \(C\in \mathfrak{A}^\times\).
		Then there exist \(a,b,c,d\in\K\), with
		\[
		ad-bc\ne0,
		\]
		such that
		\[
		w_G=\frac{aw_F+b}{cw_F+d}.
		\]
	\end{lemma}
	\begin{proof}
		By Lemma~\ref{lem:degree-of-wH}, both \(w_F\) and \(w_G\) have degree
		\(s/2\) over \(\K\).  Set
		\[
		t:=\frac{s}{2}.
		\]
		Since
		\[
		n<\ell<\frac{ns}{4},
		\]
		we have
		\[
		1\le a_0=\left\lfloor\frac{\ell}{n}\right\rfloor<\frac{s}{4}=\frac t2.
		\]
		In particular,
		\[
		2a_0<t.
		\]
		Since \(w_F\) has degree \(t\) over \(\K\), the elements
		\[
		1,w_F,\ldots,w_F^{t-1}
		\]
		form a \(\K\)-basis of \(\F\).  After extension of scalars from \(\K\) to
		\(\LL\), they form an \(\LL\)-basis of
		\[
		\F\otimes_{\K}\LL.
		\]
		Hence evaluation at \(w_F\) is injective on the set of polynomials in
		\(\LL[X]\) of degree at most \(2a_0\).  The same holds for \(w_G\).
		
		By the assumed equality of the two \(\LL\)-subspaces, for every
		\(i=0,\ldots,a_0\) there exists a polynomial
		\[
		B_i(X)\in\LL[X],
		\qquad
		\deg B_i\le a_0,
		\]
		such that
		\[
		w_G^i=B_i(w_F)C.
		\]
		For \(i=0\), this gives
		\[
		1=B_0(w_F)C.
		\]
		Thus \(B_0(w_F)\) is invertible and
		\[
		C=B_0(w_F)^{-1}.
		\]
		In particular,
		\[
		w_G=\frac{B_1(w_F)}{B_0(w_F)}.
		\]
		
		For \(i=1,\ldots,a_0\), using
		\[
		w_G^i=w_G^{i-1}w_G,
		\]
		we obtain
		\[
		B_i(w_F)C
		=
		B_{i-1}(w_F)C\,B_1(w_F)C.
		\]
		Multiplying by \(C^{-1}=B_0(w_F)\), we get
		\[
		B_i(w_F)B_0(w_F)
		=
		B_{i-1}(w_F)B_1(w_F).
		\]
		Both sides are evaluations at \(w_F\) of polynomials of degree at most
		\(2a_0\).  Since
		\[
		2a_0<t,
		\]
		evaluation at \(w_F\) is injective in this degree range.  Therefore
		\[
		B_i(X)B_0(X)=B_{i-1}(X)B_1(X),
		\qquad
		1\le i\le a_0,
		\]
		as polynomial identities in \(\LL[X]\).
		Set
		\[
		k:=a_0.
		\]
		Since \(1\) and \(w_G\) are linearly independent over \(\LL\), we have
		$
		w_G\notin\LL$. Moreover, each polynomial \(B_i\) is
		nonzero: indeed, \(B_i(w_F)C=w_G^i\neq0\), and evaluation at \(w_F\)
		is injective for polynomials of degree at most \(k\).
		Suppose first that \(k=1\).  Write
		\[
		B_1(X)=a_1X+b_1
		\qquad\text{and}\qquad
		B_0(X)=c_1X+d_1,
		\]
		with \(a_1,b_1,c_1,d_1\in\LL\).  Since \(B_0(w_F)\) is invertible, we
		obtain
		\[
		w_G
		=
		\frac{B_1(w_F)}{B_0(w_F)}
		=
		\frac{a_1w_F+b_1}{c_1w_F+d_1}.
		\]
		Furthermore, \(B_1\) and \(B_0\) cannot be proportional over \(\LL\),
		since otherwise \(w_G\in\LL\).  Therefore,
		\[
		a_1d_1-b_1c_1\neq0.
		\]
		
		Assume now that \(k>1\).  We distinguish two cases.
		
		\noindent
		\emph{Case 1: \(B_1\) does not divide \(B_0\).}
		Since \(B_1\nmid B_0\), there exist an irreducible polynomial
		\(P\in\LL[X]\), an integer \(r\geq 1\), and a polynomial
		\(G_1\in\LL[X]\) such that
		\[
		B_1=G_1P^r,
		\qquad
		\gcd(G_1,P)=1,
		\]
		and
		\[
		P^r\nmid B_0.
		\]
		
		We claim that
		\[
		P^{r+i-1}\mid B_i
		\]
		for every \(i=1,\ldots,k\). For \(i=1\), this follows immediately from
		\[
		B_1=G_1P^r.
		\]
		Assume that
		\[
		P^{r+i-2}\mid B_{i-1}
		\]
		for some \(i\in\{2,\ldots,k\}\). From
		\[
		B_iB_0=B_{i-1}B_1
		\]
		it follows that
		\[
		P^{2r+i-2}\mid B_iB_0.
		\]
		Since \(P^r\nmid B_0\), at most \(r-1\) factors equal to \(P\) divide
		\(B_0\). Therefore,
		\[
		P^{r+i-1}\mid B_i.
		\]
		This proves the claim by induction.
		In particular,
		\[
		P^{r+k-1}\mid B_k.
		\]
		Since \(B_k\neq 0\) and \(\deg B_k\leq k\), we obtain
		\[
		k
		\geq
		\deg B_k
		\geq
		(r+k-1)\deg P.
		\]
		
		Since \(r\geq 1\) and \(\deg P\geq 1\), this forces
		\[
		r=1
		\qquad\text{and}\qquad
		\deg P=1.
		\]
		Moreover, \(P^k\mid B_k\) and \(\deg B_k\leq k\), so
		\[
		B_k=\lambda P^k
		\]
		for some \(\lambda\in\LL^\times\).
		We now have
		\[
		B_1=G_1P,
		\qquad
		\gcd(G_1,P)=1.
		\]
		Using the identity
		\[
		B_kB_0=B_{k-1}B_1,
		\]
		we obtain
		\[
		\lambda P^kB_0=B_{k-1}G_1P.
		\]
		After cancelling \(P\), this becomes
		\[
		\lambda P^{k-1}B_0=B_{k-1}G_1.
		\]
		Thus \(G_1\) divides \(P^{k-1}B_0\). Since
		\[
		\gcd(G_1,P)=1,
		\]
		it follows that \(G_1\mid B_0\). Set
		\[
		E:=\frac{B_0}{G_1}.
		\]
		Then
		\[
		B_{k-1}=\lambda P^{k-1}E.
		\]
		Since \(\deg B_{k-1}\leq k\) and \(\deg P=1\), we conclude that
		\[
		\deg E\leq 1.
		\]
		Furthermore,
		\[
		B_0(w_F)=G_1(w_F)E(w_F).
		\]
		Therefore,
		\[
		w_G
		=
		\frac{B_1(w_F)}{B_0(w_F)}
		=
		\frac{P(w_F)}{E(w_F)}.
		\]
		Write
		\[
		P(X)=a_1X+b_1
		\qquad\text{and}\qquad
		E(X)=c_1X+d_1,
		\]
		with \(a_1,b_1,c_1,d_1\in\LL\). We obtain
		\[
		w_G
		=
		\frac{a_1w_F+b_1}{c_1w_F+d_1}.
		\]
		Finally, \(P\) and \(E\) cannot be proportional over \(\LL\), since
		otherwise their quotient would belong to \(\LL\), and hence
		\(w_G\in\LL\), a contradiction. Consequently,
		\[
		a_1d_1-b_1c_1\neq 0.
		\]
		
		\noindent
		\emph{Case 2: \(B_1\) divides \(B_0\).}
		Write
		\[
		B_0=B_1E
		\]
		for some \(E\in\LL[X]\).  Since
		\[
		w_G
		=
		\frac{B_1(w_F)}{B_0(w_F)}
		=
		E(w_F)^{-1},
		\]
		the polynomial \(E\) cannot be constant, as otherwise \(w_G\in\LL\).
		Thus
		\[
		\deg E\geq1.
		\]
		Substituting \(B_0=B_1E\) into
		\[
		B_iB_0=B_{i-1}B_1
		\]
		and cancelling the nonzero polynomial \(B_1\), we obtain
		\[
		B_iE=B_{i-1},
		\qquad
		1\leq i\leq k.
		\]
		It follows inductively that
		\[
		B_i=B_kE^{k-i},
		\qquad
		0\leq i\leq k.
		\]
		In particular,
		\[
		B_0=B_kE^k.
		\]
		Therefore,
		\[
		k
		\geq
		\deg B_0
		=
		\deg B_k+k\deg E.
		\]
		
		Since \(B_k\neq0\) and \(\deg E\geq1\), we conclude that
		\[
		B_k\in\LL^\times
		\qquad\text{and}\qquad
		\deg E=1.
		\]
		Write
		\[
		E(X)=c_1X+d_1,
		\qquad
		c_1\neq0.
		\]
		Then
		\[
		w_G
		=
		\frac{1}{c_1w_F+d_1}
		=
		\frac{0w_F+1}{c_1w_F+d_1}.
		\]
		Thus, taking
		\[
		a_1=0
		\qquad\text{and}\qquad
		b_1=1,
		\]
		we have
		\[
		a_1d_1-b_1c_1=-c_1\neq0.
		\]
		
		In all cases, there exist \(a_1,b_1,c_1,d_1\in\LL\), with
		\[
		a_1d_1-b_1c_1\neq0,
		\]
		such that
		\[
		w_G=\frac{a_1w_F+b_1}{c_1w_F+d_1}.
		\]
		Equivalently,
		\[
		c_1w_Fw_G+d_1w_G-a_1w_F-b_1=0
		\]
		inside \(\F\otimes_{\K}\LL\).
		
		It remains to descend the coefficients from \(\LL\) to \(\K\).  The last
		identity gives an \(\LL\)-linear dependence among the four elements
		\[
		1,\quad w_F,\quad w_G,\quad w_Fw_G
		\]
		of the \(\K\)-vector space \(\F\).  Since
		\[
		\F\otimes_{\K}\LL
		\]
		is obtained from \(\F\) by extension of scalars, these four elements are
		\(\LL\)-linearly dependent after extension of scalars if and only if they are
		\(\K\)-linearly dependent before extension of scalars.  Hence there exist
		\(a,b,c,d\in\K\), not all zero, such that
		\[
		cw_Fw_G+dw_G-aw_F-b=0.
		\]
		Equivalently,
		\[
		(cw_F+d)w_G=aw_F+b.
		\]
		We claim that
		\[
		cw_F+d\ne0.
		\]
		Indeed, if \(cw_F+d=0\), then either \(c=0=d\), or \(w_F\in\K\).  The latter is
		impossible because \(w_F\) has degree \(t=s/2>1\) over \(\K\).  If
		\(c=d=0\), then the relation above gives
		\[
		aw_F+b=0,
		\]
		again forcing \(w_F\in\K\), unless \(a=b=0\), which is excluded.  Thus
		\[
		cw_F+d\ne0.
		\]
		Since \(\F\) is a field, \(cw_F+d\) is invertible.  Therefore
		\[
		w_G=\frac{aw_F+b}{cw_F+d}.
		\]
		
		Finally, the matrix
		\[
		\begin{pmatrix}
			a & b\\
			c & d
		\end{pmatrix}
		\]
		is nonsingular.  If it were singular, then the two linear polynomials
		\[
		aX+b
		\qquad\text{and}\qquad
		cX+d
		\]
		would be proportional over \(\K\), and the quotient
		\[
		\frac{aX+b}{cX+d}
		\]
		would be constant in \(\K\).  Hence the equality above would force
		\[
		w_G\in\K,
		\]
		contradicting the fact that \(w_G\) has degree \(t=s/2>1\) over \(\K\).
		Thus
		\[
		ad-bc\ne0.
		\]
		This proves the assertion.
	\end{proof}

 We now isolate the semilinearity component needed to detect it.
	
	For \(r\in\mathbb Z/n\mathbb Z\), define
	\[
	\mathcal A_r
	:=
	\{A\in\End_{\E}(V): A\alpha=\sigma^r(\alpha)A
	\text{ for every }\alpha\in\LL\}.
	\]
	Then
	\[
	\End_{\E}(V)
	=
	\bigoplus_{r\in\mathbb Z/n\mathbb Z}\mathcal A_r,
	\]
	see, e.g., \cite[Lemma~2.3]{Dempwolff2011autotopism}.
	The component \(\mathcal A_0\) consists precisely of the operators commuting
	with \(\LL\).
	
	\begin{lemma}\label{lem:central-component-Hermitian-band}
		Let \(H\in\{F,G\}\), and put
		\[
		a:=\left\lfloor\frac{\ell}{n}\right\rfloor.
		\]
		Then
		\[
		\mathfrak C_\ell(H)\cap\mathcal A_0
		=
		\left\langle
		1,w_H,w_H^2,\ldots,w_H^a
		\right\rangle_{\LL}.
		\]
	\end{lemma}
	
	\begin{proof}
		From the explicit expression
		\[
		\mathfrak C_\ell(H)
		=
		\left\{
		\alpha_0+
		\sum_{i=1}^{\ell}
		\left(
		\sigma^{-i}(\alpha_i)T_H^{-i}
		+
		\alpha_iT_H^i
		\right)
		:
		\alpha_0,\ldots,\alpha_\ell\in\LL
		\right\},
		\]
		the terms lying in \(\mathcal A_0\) are precisely those with exponent
		divisible by \(n\). Indeed, for every integer \(j \geq 0\), one has
	\[
	T_H^j\in\mathcal A_{j}.
	\] Since
	$
	\End_{\E}(V)
	=
	\bigoplus_{r\in\mathbb Z/n\mathbb Z}\mathcal A_r,
	$
	an element of \(\mathfrak C_\ell(H)\) belongs to \(\mathcal A_0\)
	only if all its components in \(\mathcal A_r\), with \(r\neq0\),
	vanish. This
	implies that
	\[
	\alpha_i=0
	\qquad
	\text{whenever }n\nmid i.
	\]
	Thus the only possible nonzero indices are
	\[
	j=t n,
	\qquad
	t=1,\ldots,
	a.
	\]
	For these indices,
	\[
	\sigma^{-j}=\id_{\LL},
	\qquad
	T_H^{j}=u_H^t\id_V,
	\qquad
	T_H^{-j}=u_H^{-t}\id_V.
	\]
	Hence
		\[
		\mathfrak C_\ell(H)\cap\mathcal A_0
		=
		\left\langle
		1,\,
		u_H+u_H^{-1},\,
		u_H^2+u_H^{-2},\,
		\ldots,\,
		u_H^a+u_H^{-a}
		\right\rangle_{\LL}.
		\]
		Now put
		\[
		w_H=u_H+u_H^{-1}.
		\]
		For every \(j\ge1\), the element
		\[
		u_H^j+u_H^{-j}
		\]
		is a monic polynomial of degree \(j\) in \(w_H\).  Indeed,
		\[
		u_H^{j+1}+u_H^{-(j+1)}
		=
		w_H(u_H^j+u_H^{-j})
		-
		(u_H^{j-1}+u_H^{-(j-1)}).
		\]
		Therefore
		\[
		\left\langle
		1,\,
		u_H+u_H^{-1},\,
		u_H^2+u_H^{-2},\,
		\ldots,\,
		u_H^a+u_H^{-a}
		\right\rangle_{\LL}
		=
		\left\langle
		1,w_H,w_H^2,\ldots,w_H^a
		\right\rangle_{\LL}.
		\]
		This proves the claim.
	\end{proof}
	
	We also need the following simple consequence of the Lyapunov computation.
	\begin{lemma}\label{lem:both-sided-normalization}
	Let \(Q\) and \(C\) be as in
	Lemma~\ref{lem:common-equivalence-reduction}.  Then the map
	\[
	A\longmapsto QAQ^{-1}C
	\]
	preserves each \(\mathcal A_r\).
	\end{lemma}
	
	\begin{proof}
	Since
	\[
	Q\LL Q^{-1}=\LL,
	\]
	there exists \(h\in\mathbb Z/n\mathbb Z\) such that
	\[
	Q\alpha Q^{-1}
	=
	\sigma^h(\alpha)
	\qquad
	\text{for every }\alpha\in\LL.
	\]
	Equivalently,
	\[
	Q^{-1}\alpha
	=
	\sigma^{-h}(\alpha)Q^{-1}.
	\]
	Let \(A\in\mathcal A_r\).  Then, for every \(\alpha\in\LL\),
	\[
	\begin{aligned}
	(QAQ^{-1})\alpha
	&=
	QA\sigma^{-h}(\alpha)Q^{-1}\\
	&=
	Q\sigma^{r-h}(\alpha)AQ^{-1}\\
	&=
	\sigma^r(\alpha)QAQ^{-1}.
	\end{aligned}
	\]
	Hence
	\[
	QAQ^{-1}\in\mathcal A_r.
	\]
	Moreover,
	\[
	C\in \Cen_{\End_\E(V)}(\LL)=\mathcal A_0.
	\]
	Therefore right multiplication by \(C\) preserves every
	\(\mathcal A_r\).  Consequently,
	\[
	QAQ^{-1}C\in\mathcal A_r
	\]
	for every \(A\in\mathcal A_r\).  Since the map is bijective, it
	preserves each component \(\mathcal A_r\).
	\end{proof}
    
	We are now ready to prove the large-range main equivalence theorem stated at
	the beginning of this section.
	\begin{proof}[Proof of Theorem~\ref{thm:Hermitian-equivalence-PGL2-large-ell}]
		Assume that \(\newH_\ell(F)\) and \(\newH_\ell(G)\) are linearly
		Hermitian equivalent. By Lemma~\ref{lem:common-equivalence-reduction}, there exist
	\[
	Q\in\GL_\E(V)
	\qquad\text{and}\qquad
	C\in \Cen_{\End_\E(V)}(\LL)^\times
	\]
	such that
	\[
	\mathfrak C_\ell(G)
	=
	Q\mathfrak C_\ell(F)Q^{-1}C
	\]
	and
	\[
	Q\LL Q^{-1}=\LL.
	\]
		
		By Lemma~\ref{lem:both-sided-normalization},  the map
	\[
	A\longmapsto QAQ^{-1}C
	\]
	preserves \(\mathcal A_0\). Consequently, we have
	\[
	\mathfrak C_\ell(G)\cap\mathcal A_0
	=
	Q
	\left(
	\mathfrak C_\ell(F)\cap\mathcal A_0
	\right)
	Q^{-1}C.
	\]
	Put
	\[
		a_0:=\left\lfloor\frac{\ell}{n}\right\rfloor.
	\]
	By Lemma~\ref{lem:central-component-Hermitian-band}, we get
	\[
	\left\langle
	1,w_G,\ldots,w_G^{a_0}
	\right\rangle_{\LL}
	=
	Q
	\left\langle
	1,w_F,\ldots,w_F^{a_0}
	\right\rangle_{\LL}
	Q^{-1}C.
	\]
	Since \(Q\) normalizes \(\LL\) and fixes the scalar field \(\E\)
	pointwise under conjugation, one has
	\[
	Q
	\left\langle
	1,w_F,\ldots,w_F^{a_0}
	\right\rangle_{\LL}
	Q^{-1}
	=
	\left\langle
	1,w_F,\ldots,w_F^{a_0}
	\right\rangle_{\LL}.
	\]
	Therefore
	\[
	\left\langle
	1,w_G,\ldots,w_G^{a_0}
	\right\rangle_{\LL}
	=
	\left\langle
	1,w_F,\ldots,w_F^{a_0}
	\right\rangle_{\LL}C.
	\]
			
		The element \(C\) is invertible.  Moreover, since \(1\) belongs to
		\[
		\left\langle
		1,w_F,\ldots,w_F^{a_0}
		\right\rangle_{\LL},
		\]
		the equality above implies
		\[
		C\in
		\left\langle
		1,w_G,\ldots,w_G^{a_0}
		\right\rangle_{\LL}
		\subseteq
		\F\otimes_{\K}\LL.
		\]
		Thus the equality takes place inside the commutative algebra
		\[
		\F\otimes_{\K}\LL.
		\]
		
		Now, by Lemma~\ref{lem:degree-of-wH}, both \(w_F\) and \(w_G\) have degree
		\[
		t:=\frac{s}{2}
		\]
		over \(\K\).  Since
		\[
		n<\ell<\frac{ns}{4},
		\]
		we have
		\[
		1\le a_0=\left\lfloor\frac{\ell}{n}\right\rfloor<\frac{s}{4}=\frac t2.
		\]
		Therefore Lemma~\ref{lem:central-interval-determination}, applied with
		\[
		u=w_F,
		\qquad
		v=w_G,
		\qquad
		m=a_0,
		\]
		gives
		\[
		w_G=\frac{aw_F+b}{cw_F+d}
		\]
		for some \(a,b,c,d\in\K\), with
		\[
		ad-bc\ne0. \qedhere
		\]
	\end{proof}
	
	\begin{remark}
		The conclusion of Theorem~\ref{thm:Hermitian-equivalence-PGL2-large-ell} is
		weaker than the conclusion of
		Theorem~\ref{thm:Hermitian-equivalence-polynomial-condition}.  In the short
		range \(\ell<n\), the interval $[-\ell,\ell]$ can be recovered after taking the
		\(\LL\)-span, and one obtains the stronger relation
		\[
		F(Y)=\lambda^sG(\lambda^{-1}Y)
		\]
		for some \(\lambda\in\K^\times\).  In the range \(n<\ell<ns/4\), the
		exponents divisible by \(n\) occur only through the symmetric combinations
		\[
		u_H^j+u_H^{-j}.
		\]
		Thus the natural invariant is not the \(\mathrm{PGL}_2(\K)\)-orbit of
		\(u_H\), but rather the \(\mathrm{PGL}_2(\K)\)-orbit of
		\[
		w_H=u_H+u_H^{-1}.
		\]
	\end{remark}

	\subsection{Enumerative consequences}
	\label{subsec:enumerative-consequences}
	
	The complete bounds for full Hermitian equivalence were collected in
	Corollary~\ref{cor:enumerative-Hermitian-combined}.  We now state and prove the
	sharper linearly Hermitian bounds in this subsection.

	Recall that
	\[
		\ell_d=\frac{s(n-d+1)-2}{2},
	\]
	and 
	\[
		\Lambda_{s,d}:=\left\{
		\newH_{\ell_d}(F):
		F\in\mathcal I^{\mathrm{sr}}_{q,s}
		\right\}.
	\]
	For the short range \(\ell_d<n\), the equivalence among members in $\Lambda_{s,d}$ is controlled by the scaling relation in Theorem~\ref{thm:Hermitian-equivalence-polynomial-condition}, whereas for the large
	range \(n<\ell_d<ns/4\) the equivalence is controlled by the self-reciprocal \(\mathrm{PGL}_2(\F_q)\)-relation in
	Theorem~\ref{thm:Hermitian-equivalence-PGL2-large-ell}.
	
	We shall use the following elementary observation several times.  If
	\(F\in\mathcal I^{\mathrm{sr}}_{q,s}\), then
	\[
	F(0)=1.
	\]
	Indeed, since \(F\) is irreducible of degree \(s>1\), neither \(1\) nor \(-1\)
	is a root of \(F\).  The roots of a self-reciprocal polynomial are stable under
	\(\alpha\mapsto \alpha^{-1}\), and therefore they split into inverse pairs.
	Since \(s\) is even, the product of all roots is \(1\), and hence the constant
	coefficient of the monic polynomial \(F\) is equal to \(1\).
	
	We first treat separately the quadratic case \(s=2\).
	
	\begin{corollary}\label{cor:enumerative-s-equal-2-linear}
		Assume that \(s=2\), \(n>2\), and
		\[
		\frac n2<d<n.
		\]
		Then $\Lambda_{s,d}$
		contains at least
		\[
		L_2(q):=
		\begin{cases}
			\dfrac q2, & q \textup{ even},\\[2ex]
			\dfrac{q-1}{4}, & q\equiv 1 \pmod 4,\\[2ex]
			\dfrac{q+1}{4}, & q\equiv 3 \pmod 4,
		\end{cases}
		\]
		pairwise linearly Hermitian inequivalent maximum additive Hermitian
		\(d\)-codes in \(\mathcal H_n(\F_{q^2}/\F_q)\).
	\end{corollary}
	
	\begin{proof}
		For \(s=2\), one has
		\[
		\ell=n-d.
		\]
		The short-range condition
		\[
		1\le \ell<\frac n2
		\]
		is therefore equivalent to
		\[
		\frac n2<d<n.
		\]
		
		Every monic self-reciprocal polynomial of degree \(2\) has the form
		\[
		F_c(Y)=Y^2+cY+1,
		\qquad c\in\F_q.
		\]
		By \eqref{eq:Carlitz-formula}, we get
		\[
		|\mathcal I^{\mathrm{sr}}_{q,2}|=N_q(2)=
		\begin{cases}
			(q-1)/2, & q \textup{ odd},\\
			q/2, & q \textup{ even}.
		\end{cases}
		\]
		
		Let \(F_c,F_{c'}\in\mathcal I^{\mathrm{sr}}_{q,2}\).  If
		\(\newH_{n-d}(F_c)\) and \(\newH_{n-d}(F_{c'})\) are linearly
		Hermitian equivalent, then, by
		Theorem~\ref{thm:Hermitian-equivalence-polynomial-condition}, there exists
		\(\lambda\in\F_q^\times\) such that
		\[
		F_c(Y)=\lambda^2F_{c'}(\lambda^{-1}Y).
		\]
		Since
		\[
		\lambda^2F_{c'}(\lambda^{-1}Y)
		=
		Y^2+\lambda c'Y+\lambda^2,
		\]
		comparison of the constant terms gives
		\[
		\lambda^2=1.
		\]
		
		If \(q\) is even, then \(\lambda=1\), and hence \(c=c'\).  Therefore distinct
		polynomials in \(\mathcal I^{\mathrm{sr}}_{q,2}\) yield pairwise linearly
		Hermitian inequivalent codes.  This gives \(q/2\) such codes.
		
		Assume now that \(q\) is odd.  Then \(\lambda=\pm1\), so \(F_c\) can only be
		identified with \(F_c\) or with \(F_{-c}\).  Thus it is enough to count the
		orbits of the involution
		\[
		c\longmapsto -c
		\]
		on the set of parameters \(c\in\F_q\) for which \(F_c(Y)\) is irreducible.
		This set has cardinality \((q-1)/2\).  The only possible fixed point is
		\(c=0\), which belongs to the set precisely when \(Y^2+1\) is irreducible over
		\(\F_q\), equivalently when \(-1\) is not a square in \(\F_q\).  This happens
		exactly when \(q\equiv 3\pmod 4\).  Hence the number of orbits is
		\[
		\frac{(q-1)/2}{2}
		=
		\frac{q-1}{4}
		\]
		if \(q\equiv1\pmod4\), and
		\[
		\frac{(q-1)/2+1}{2}
		=
		\frac{q+1}{4}
		\]
		if \(q\equiv3\pmod4\).  This proves the claim.
	\end{proof}
	

	We now turn to the case \(s\ge4\).  We first consider the short range
	\[
	\ell<n.
	\]
	In terms of \(d\), this condition is
	\begin{equation}\label{eq:enumerative-range-short}
		n>2,
		\qquad
		\frac{s(n-d+1)-2}{2}<n.
	\end{equation}
	For \(s\ge4\) and \(d\le n\), the lower bound
	\[
	1\le \frac{s(n-d+1)-2}{2}
	\]
	is automatic.
	
	\begin{theorem}\label{thm:enumerative-linear-s-geq-4-short}
		Assume that \(s\ge4\) is even, \(n>2\), and
		\[
		\ell=\frac{s(n-d+1)-2}{2}<n.
		\]
		Then the family
		$\Lambda_{s,d}$
		contains at least
		\[
		\left\lceil
		\frac{N_q(s)}{\gcd(s,q-1)}
		\right\rceil
		\]
		pairwise linearly Hermitian inequivalent maximum additive Hermitian
		\(d\)-codes in \(\mathcal H_n(\F_{q^s}/\F_{q^{s/2}})\).
	\end{theorem}
	
	\begin{proof}
		Let \(F,G\in\mathcal I^{\mathrm{sr}}_{q,s}\).  Suppose that
		\(\newH_\ell(F)\) and \(\newH_\ell(G)\) are linearly Hermitian
		equivalent.  By Theorem~\ref{thm:Hermitian-equivalence-polynomial-condition},
		there exists \(\lambda\in\F_q^\times\) such that
		\[
		F(Y)=\lambda^sG(\lambda^{-1}Y).
		\]
		Since \(F(0)=G(0)=1\), comparison of the constant terms gives
		\[
		1=F(0)=\lambda^sG(0)=\lambda^s.
		\]
		Thus \(\lambda\in\F_q^\times\) must satisfy
		\[
		\lambda^s=1.
		\]
		For each fixed \(G\), there are therefore at most
		\[
		\gcd(s,q-1)
		\]
		possible polynomials \(F\in\mathcal I^{\mathrm{sr}}_{q,s}\) which can give
		rise to a code linearly Hermitian equivalent to \(\newH_\ell(G)\).
		Since
		\[
		|\mathcal I^{\mathrm{sr}}_{q,s}|=N_q(s),
		\]
		the asserted lower bound follows.
	\end{proof}
	
		
	We now consider the larger range
	\[
	n<\ell<\frac{ns}{4}.
	\]
	This range is nonempty only when \(s>4\).  In terms of \(d\), the required
	condition is
	\begin{equation}\label{eq:enumerative-range-large}
		n>2,
		\qquad
		n<
		\frac{s(n-d+1)-2}{2}
		<
		\frac{ns}{4}.
	\end{equation}
	
	\begin{theorem}\label{thm:enumerative-linear-large-ell-PGL2}
		Assume that \(s>4\) is even, \(n>2\), and
		\[
		n<\ell=
		\frac{s(n-d+1)-2}{2}
		<
		\frac{ns}{4}.
		\]
		Then $\Lambda_{s,d}$
		contains at least
		\[
		\left\lceil
		\frac{N_q(s)}{q(q^2-1)}
		\right\rceil
		\]
		pairwise linearly Hermitian inequivalent maximum additive Hermitian
		\(d\)-codes in \(\mathcal H_n(\F_{q^s}/\F_{q^{s/2}})\).
	\end{theorem}
	
	\begin{proof}
		Let \(F,G\in\mathcal I^{\mathrm{sr}}_{q,s}\).  Suppose that
		\(\newH_\ell(F)\) and \(\newH_\ell(G)\) are linearly Hermitian
		equivalent.  Since
		\[
		n<\ell<\frac{ns}{4},
		\]
		Theorem~\ref{thm:Hermitian-equivalence-PGL2-large-ell} provides the existence of $a,b,c$ and $d\in \K$ with $ad-bc\neq 0$ such that
		\[
		w_G=\frac{aw_F+b}{cw_F+d}.
		\]
		Thus every linear Hermitian equivalence class can contain polynomials from at
		most one self-reciprocal \(\mathrm{PGL}_2(\F_q)\)-orbit.  Each such orbit has
		cardinality at most
		\[
		|\mathrm{PGL}_2(\F_q)|=q(q^2-1).
		\]
		Therefore each linear Hermitian equivalence class contains at most
		\(q(q^2-1)\) elements of \(\mathcal I^{\mathrm{sr}}_{q,s}\).  Since
		\[
		|\mathcal I^{\mathrm{sr}}_{q,s}|=N_q(s),
		\]
		the asserted lower bound follows.
	\end{proof}
	
	\begin{proof}[Proof of Corollary~\ref{cor:enumerative-Hermitian-combined}]
		For part~\textup{(i)}, the full Hermitian equivalence relation allows field
		automorphisms of \(\F_{q^2}\).  Since \(q=p^h\),
		\[
		|\Aut(\F_{q^2})|=[\F_{q^2}:\F_p]=2h.
		\]
		Consequently, every full Hermitian equivalence class is a union of at most
		\(2h\) linearly Hermitian equivalence classes.  Applying
		Corollary~\ref{cor:enumerative-s-equal-2-linear} gives the first bound.
		
		For parts~\textup{(ii)} and~\textup{(iii)}, the relevant field is
		\(\F_{q^s}\), and
		\[
		|\Aut(\F_{q^s})|=[\F_{q^s}:\F_p]=sh.
		\]
		Thus every full Hermitian equivalence class is a union of at most \(sh\)
		linear classes.  Dividing the bounds in
		Theorem~\ref{thm:enumerative-linear-s-geq-4-short} and
		Theorem~\ref{thm:enumerative-linear-large-ell-PGL2} by this factor, and taking
		ceilings, gives parts~\textup{(ii)} and~\textup{(iii)}, respectively.
	\end{proof}
	
	Finally, \eqref{eq:Carlitz-formula} gives
	\[
		N_q(s)=\Theta(q^{s/2})  \text{ as }q\rightarrow \infty.
	\]
	Hence the linear counts have order \(q\) in the quadratic case, order
	\(q^{s/2}\) in the short range for fixed even \(s\ge4\), and order at least
	\(q^{s/2-3}\) in the large range for fixed even \(s>4\).  The corresponding
	full Hermitian counts lose only the field-automorphism factor: they have orders
	at least \(q/h\), \(q^{s/2}/h\), and \(q^{s/2-3}/h\), respectively.
	Thus our construction yields many inequivalent new maximum additive Hermitian
	rank-metric codes, in contrast with the single classical trace construction
	available for the same parameters.
	
	\section*{Acknowledgments}
	Paolo Santonastaso is partially supported by the Italian National Group for
	Algebraic and Geometric Structures and their Applications (GNSAGA--INdAM).
	Yue Zhou is partially supported by the National Natural Science Foundation of
	China (No.\ 12371337).
	
	The authors are deeply grateful to the KIMSI Institute in Changsha, China, for
	hosting them and providing office space and facilities during the visit in which
	most of this work was completed.
	
	\bibliographystyle{abbrv}
	\bibliography{biblio}
		
\end{document}